\documentclass[journal]{IEEEtran}

\usepackage{amsmath,amssymb,amsfonts}
\usepackage{graphicx}
\usepackage{enumitem}
\usepackage{color}
\usepackage{comment}
\usepackage{xcolor}
\usepackage{amsthm}
\usepackage{algorithm}
\usepackage{algpseudocode}
\usepackage{my_symbol}
\usepackage{amsmath, amssymb, amsfonts, amsthm}
\usepackage{graphicx}
\usepackage{bbm}         
\usepackage{mathtools}   
\usepackage{enumitem}    
\usepackage{bm}          
\usepackage{xcolor}
\usepackage{hyperref}
\hypersetup{
    colorlinks=true,
    linkcolor=blue,
    citecolor=blue,
    urlcolor=blue
}
\usepackage{algorithm}
\usepackage{algpseudocode}

\newtheorem{theorem}{Theorem}
\newtheorem{lemma}[theorem]{Lemma}
\newtheorem{proposition}[theorem]{Proposition}
\newtheorem{corollary}[theorem]{Corollary}
\theoremstyle{definition}
\newtheorem{definition}[theorem]{Definition}
\newtheorem{example}[theorem]{Example}
\newtheorem{remark}[theorem]{Remark}

\def\R{\mathbb{R}}

\def\v{\mathbf{v}}

\def\R{\mathbb{R}}
\def\1{\mathbf{1}}

\def\cL{\bbL}

\def\op{\mathrm{op}}

\title{Sampling Transferable Graph Neural Networks with Limited Graph Information
}

\author{Haoyu~Wang, Renyuan~Ma, Gonzalo~Mateos, and~Luana~Ruiz%
    \thanks{H. Wang and L. Ruiz are with the Department of Applied Mathematics and Statistics, Johns Hopkins University, Baltimore, MD 21218 USA (e-mail: hwang320@jh.edu; lrubini1@jh.edu).}%
    \thanks{R. Ma is with the Department of Statistics and Data Science, Yale University, New Haven, CT 06520 USA (e-mail: jack.ma.rm2545@yale.edu).}%
    \thanks{G. Mateos is with the Department of Electrical and Computer Engineering, University of Rochester, Rochester, NY 14627 USA (e-mail: gmateosb@ece.rochester.edu).}%
}

\usepackage{hyperref}

\begin{document}
\maketitle

\begin{abstract}
Graph neural networks (GNNs) achieve strong performance on a wide range of graph learning tasks, but training on large-scale networks remains computationally challenging. Transferability results show that GNNs with fixed weights can generalize from smaller graphs to larger ones drawn from the same family, motivating the use of sampled subgraphs to boost training efficiency. Yet most existing sampling strategies rely on reliable access to the target (inference time) graph structure, which in practice may be noisy, incomplete, or unavailable prior to training. In lieu of precise connectivity information we study feature-driven subgraph sampling for transferable GNNs, with the goal of preserving spectral properties of graph operators that control the expressivity of the 
GNN. We adopt a fresh alignment-based perspective that links node feature statistics to graph spectral structure and develop two complementary notions of feature–graph alignment. For \textit{coarse alignment}, we formalize feature homophily through a Laplacian-based measure that quantifies the alignment of feature principal components with graph eigenvectors, and establish a lower bound on the trace of the graph Laplacian in terms of feature statistics. This motivates a simple, non-sequential sampling algorithm that operates directly on the feature matrix and preserves a trace-based proxy for operator rank, thus not compromising GNN expressivity.
For fine alignment, we assume a stationary model where the feature covariance and Laplacian share an eigenbasis. We establish that, with high probability, diagonal covariance entries reflect node-degree ordering under monotone filters. We empirically validate that the filter’s monotonicity dictates the relationship between feature variance and spectral energy. On real-world benchmarks, we find that a single retention rule typically dominates across all sampling budgets; selecting the rule that maximizes the Laplacian trace consistently yields GNNs with superior transferability and reduced generalization gaps.


\end{abstract}


\section{Introduction}
\label{sec:intro}

Graph neural networks (GNNs) are deep learning models 
tailored to network data, which have shown great empirical performance in several graph machine learning tasks; see e.g., \cite{gori2005new,kipf17-classifgcnn,defferrard17-cnngraphs,gama18-gnnarchit}. Noteworthy are graph signal processing (GSP) problems--such as recommender systems on product similarity networks \cite{ruiz2020gnns}, or, attribution of research papers to scientific domains \cite{hamilton2017inductive}--in which GNNs' invariance and stability properties \cite{ruiz19-inv,gama19-stability} play a key role. 

Yet, in practice most successful applications of GNNs are limited to graphs of moderate size. The sheer size of many modern networks, typically in the order of several millions, frequently makes these models impractical to train. Favorable results have been seen by leveraging the GNN's transferability property \cite{ruiz20-transf,levie2019transferability}, which states that a GNN with fixed weights produces similar outputs on large enough graphs belonging to the same ``family'', e.g., the same random graph model. This property justifies affordably training the GNN on a moderate-sized graph, and transferring it for inference on the large graph.

The transferability of GNNs is closely related to their convolutional parametrization (Section \ref{ssec:conv_gnn}), and is a consequence of the fact that graph convolutions converge on sequences of graphs converging to a common graph limit \cite{ruiz2020graphonsp}. Under certain assumptions on the type of limit, and on how the graphs converge to (or are sampled from) them, it is possible to obtain non-asymptotic error bounds inversely proportional to the sizes of the graphs. Such bounds are then used to inform the minimum graph size on which to train a GNN to meet a maximum transference error. 
Once this is determined, the training graph is obtained by sampling a subgraph with the prescribed size. However, most existing sampling approaches implicitly assume reliable access to the graph structure. In many practical settings, this assumption does not hold, e.g., the observed graph may be noisy, partially missing, or temporally outdated. Accordingly, graph-based sampling procedures may be infeasible, motivating the need for sampling strategies that do not depend critically on the observed graph. In this challenging context, we study sampling for transferable GNNs under limited graph information, with the goal of preserving spectral properties of graph operators that control the expressivity of the
GNN; see also Section \ref{subsec:prob_statement}.\vspace{2pt}

\noindent\textbf{Related work and our distinct focus.} Node features provide a natural alternative source of information. This perspective is well established in GSP, wherein node features can be modeled as graph signals and used to infer structural properties. Examples include smoothness-based learning frameworks that use feature variation to characterize spectral content \cite{shuman2013gsp,belkin2006manifold,zhu2003gaussian}. Feature covariance matrices have been used for topology identification under graph stationarity or low-rank assumptions \cite{segarra2017topology,dong2016laplacian,kalofolias2016learn}, for global network property estimation from node signals (e.g., centrality~\cite{roddenberry2021blind,he2022detecting}), and more recently for the design of covariance neural networks (VNNs)--GNNs supported on feature covariance graphs~\cite{VNN_NeurIPS_22,VNN_JSTSP,VNN_SPMAG_2}. In lieu of precise network connectivity information, we also leverage feature correlations as proxies of structure, but our novel angle is to guide graph sampling for transferable GNNs.

When the graph is available, sampling strategies are typically designed to preserve spectral properties of graph operators that control the expressivity of graph filters, such as eigenvectors up to a given bandwidth, operator rank, or trace-based proxies. These objectives are central to classical sampling results in GSP and optimal experimental design \cite{anis2016spectralproxies,tsitsvero2016uncertainty,pukelsheim2006optimal}, and  directly relate to the stability and transferability of GNNs \cite{gama2020stability}, whose convolutional layers are parametrized as graph filters. Unlike~\cite{anis2016spectralproxies,tsitsvero2016uncertainty} that deal with graph signal reconstruction, here we study graph sampling as a vehicle for GNN transferability and efficient training. The challenge in the feature-driven setting is to recover analogous sampling criteria without direct access to the graph operator itself.\vspace{2pt}

\noindent\textbf{Proposed approach, contributions, and paper outline.} Our novel idea is to adopt an alignment-based perspective that links feature statistics to spectral properties of the underlying graph operator. The guiding principle is that, when node features are sufficiently aligned with the graph structure, spectral proxies relevant to graph filtering can be recovered directly from feature information.
We study two complementary notions of feature–graph alignment dubbed \textit{coarse} and \textit{fine}. 

At the \textit{coarse} level, in Section \ref{sec:coarse_alignment} we formalize feature homophily through a Laplacian-based measure that quantifies how feature principal components align with low- or high-frequency eigenvectors of the graph Laplacian. Small values of this measure indicate that features concentrate their energy on low-frequency eigenvectors, which encode global graph structure such as communities. Under this condition, we establish a lower bound on the trace of the graph Laplacian in terms of feature statistics, showing that minimizing the trace of the feature correlation matrix leads to improved preservation of the trace proxy for operator rank. This result motivates a simple, non-sequential sampling algorithm that operates directly on the feature matrix, incurs lower complexity than rank-preserving spectral methods, and scales to large graphs.

At the \textit{fine} level (Section \ref{sec:fine_alignment}), we assume a stationary graph signal model in which node features are generated by convolutional filtering of latent excitation signals. Under stationarity, the feature covariance matrix and the graph Laplacian commute, implying a shared eigenbasis. We also establish that diagonal entries of the feature covariance matrix reflect node degree ordering with high probability, up to finite-sample deviations. This provides a theoretical justification for using feature-only statistics to recover fine-grained structural information, and further supports feature-driven sampling when stronger alignment assumptions hold. To validate these theoretical insights, Section \ref{sec:sbm} presents a rigorous case study on the Stochastic Block Model (SBM), which provides quantifiable guarantees on degree-ordering recovery and trace preservation under monotone filters. We further discuss the implications of the underlying filter's monotonicity in Section \ref{sec:discussion}, exploring how it dictates both the optimal node sampling strategy and the nature of feature correlations on the graph. 

Relative to the conference precursor \cite{li2025graphsampling}, which primarily introduced the coarse alignment framework and the trace-minimizing sampling heuristic, this paper presents a comprehensive and unified treatment of feature-driven graph sampling. Noteworthy novel contributions in this full journal paper include: (i) the introduction of the fine alignment framework grounded in the stationary graph signal model; (ii) theoretical guarantees linking node-wise feature variance to degree ordering under monotone filters; (iii) a rigorous theoretical case study on the Stochastic Block Model (SBM) validating these spectral approximation insights; and (iv) a markedly expanded experimental evaluation encompassing both synthetic environments and a wider array of real-world datasets to assess the downstream transferability of GNNs.

\section{Preliminary Background}

An unweighted graph $G=(V,E)$ consists of a set of vertices or nodes $V$ and a set of edges $E \subseteq V \times V$. Generally, graphs can be categorized as being either directed or undirected based on their edge set $E$. A graph is undirected if and only if for any two nodes $u,v \in V$, $(u,v) \in E$ also implies $(v,u) \in E$ (and both correspond to the same undirected edge). In this paper, we restrict attention to undirected graphs.

Let $n=|V|$ be the number of nodes and $m = |E|$ be the number of edges. The $n \times n$ adjacency matrix $\mathbf{A}$ has entries
$$A_{ij} =
\begin{cases}
 1 & \text{if } (i,j) \in E \\
    0 & \text{otherwise.}
\end{cases}$$
The Laplacian matrix or graph combinatorial Laplacian is defined as $\bbL_c = \mathbf{D} - \mathbf{A}$, where $\bbD=\mbox{diag}(\bbA\boldsymbol{1}_n)$ is the so-called degree matrix. From their definitions, and since $G$ is undirected, $\mathbf{A}$ and $\bbL_c$ are symmetric. More generally, we refer to a matrix $\bbS \in \reals^{n\times n}$ that encodes the graph's sparsity pattern (i.e., $S_{ij} \neq 0$ only if $(i,j) \in E$ or $i=j$) as a \textit{graph shift operator (GSO)}. The adjacency matrix $\bbA$ and the Laplacian $\bbL_c$ are common examples of GSOs \cite{segarra17-linear}.

Real-world graphs are associated with node data $\bbx \in \reals^n$ called graph signals, where $\bbx[i]$ corresponds to the signal value at node $i$. More generally, graph signals consist of multiple features, in which case they are represented as matrices $\bbX \in \mathbb{R}^{n \times d} $ with $d$ denoting the number of features per node.

The graph Laplacian plays an important role in GSP \cite{shuman13-mag,sandryhaila13-dspg}, as it allows defining the notion of total variation of a signal $\bbx$. Explicitly, the total variation of $\bbx$ is defined as $\textrm{TV}(\bbx)=\bbx^\top\bbL_c \bbx$ \cite{ortega2018graph}. Let $\bbL_c=\bbV\bbLam\bbV^\top$ be the Laplacian eigendecomposition, where $\bbLam$ is a diagonal matrix with eigenvalues ordered as $\lambda_1 \leq \ldots \leq \lambda_n$ and $\bbV$ is the corresponding orthogonal eigenvector matrix. For unit-norm signals, the maximum total variation is $\lambda_n$ and the minimum total variation is $\lambda_1=0$, which corresponds to the all-ones eigenvector $\bbv_1=\boldsymbol{1}_n$.  Therefore, the Laplacian eigenvalues can be interpreted as graph frequencies, and the eigenvectors as these frequencies' respective oscillation modes.

\subsection{Graph Convolutions and Graph Neural Networks}\label{ssec:conv_gnn}

A graph convolutional filter is the extension of a standard convolutional filter to graph signals, defined based on a graph shift operator $\bbS$. Given any choice of $\bbS$ and a graph signal $\bbx \in \reals^n$, the graph convolution is defined as \cite{sandryhaila13-dspg,shuman13-mag}
\begin{equation} \label{eqn:simple_conv}
\bby = h(\bbS) \bbx = \sum_{k=0}^{K-1} h_k \bbS^k \bbx
\end{equation}
where $h_0, \ldots, h_{K-1}$ are the filter coefficients or taps. 

For undirected graphs $G$, let $\bbS = \bbV \bbLam \bbV^\top$ denote the eigendecomposition of $S$ and $\hbx=\bbV^\top \bbx$ and $\hby = \bbV^\top \bby$ the projections of signals $\bbx$ and $\bby$ onto the eigenvector basis. We can write \cite{shuman13-mag}
\begin{equation} \label{eqn:spectral_filter}
    \hby = h(\bbLam) \hbx = \sum_{k=0}^{K-1} h_k \bbLam^k \hbx \text{.}
\end{equation}
This is a particular case of so-called spectral graph filters \cite{hammond2011wavelets,shuman13-mag}, which are more broadly defined in terms of general measurable functions $f$ as $\hby = f(\bbLam)\hbx$.

For multi-feature $\bbX \in \reals^{n \times d}$ and $\bbY \in \reals^{n\times f}$, we can generalize \eqref{eqn:simple_conv} to a convolutional filterbank \cite{gama18-gnnarchit}
\begin{equation}
\bbY = \sum_{k=0}^{K-1} \bbS^k \bbX \bbH_k
\end{equation}
mapping features from $\reals^d$ to $\reals^f$, where $\bbH_k \in \reals^{d \times f}$, $0 \leq k \leq K-1$, group the corresponding filter coefficients.

Similarly to how image convolutions can be stacked with nonlinearities to form convolutional neural networks, GNNs are defined as sequences of layers each consisting of a graph convolution and a pointwise nonlinearity. Explicitly, we write the $\ell$th layer of a GNN as \cite{gama18-gnnarchit}
\begin{equation}
\bbX_{\ell} = \sigma \bigg( \sum_{k=0}^{K-1} \bbS^k \bbX_{\ell-1} \bbH_{\ell k}\bigg)
\end{equation}
where $\sigma: \reals \mapsto \reals$ is an entry-wise nonlinearity (e.g., the ReLU or sigmoid). At layer $\ell=1$, $\bbX_0$ is the input data $\bbX$, and the last layer output $\bbX_L$ is the output $\bbY$. For succinctness, in the following we  represent the whole $L$-layer GNN as a map $\bbY = \Phi(\bbX,G;\ccalH)$, with $\ccalH=\{\bbH_{\ell k}\}_{\ell,k}$.\vspace{2pt}

\noindent \textbf{Transferability of GNNs.} The mathematical property that allows training GNNs on small graph subsamples of the larger target graphs used for inference is their transferability. Explicitly, GNNs are transferable in the sense that when a GNN with fixed weights $\ccalH$ is transferred across two graphs in the same ``family'', the transference error is upper bounded by a term that decreases with the graph size \cite{levie2019transferability,ruiz2021transferability}. Typical transferability analyses show this by defining graph ``families'' as graphs coming from the same random graph model, or converging to a common graph limit. For example, considering families of graphs identified by the same graphon---which can be seen as both a generative model and a limit model for large graphs---, Ruiz et al. \cite{ruiz2021transferability} prove the following transferability theorem.

\begin{theorem}[GNN transferability, simplified \cite{ruiz2021transferability}]
Let $G_n$ and $G_m$ be two graphs with $n$ and $m$ nodes, respectively, sampled from a common graphon. Let $\bbX_n$ and $\bbX_m$ be the respective input graph signals. Let $\Phi(\cdot, \cdot; \ccalH)$ be a GNN with fixed coefficients $\ccalH$ (where $\ccalH$ is defined as the set of all filterbank weights $\{\bbH_{\ell k}\}_{\ell,k}$). Under mild assumptions on the graphon and the signals, the difference between the GNN outputs satisfies 
$\| {\Phi}(\bbX_n, G_n; \ccalH) - {\Phi}(\bbX_m, G_m; \ccalH) \| = \mathcal{O}\left(\sqrt{\frac{\log{n}}{n}} + \sqrt{\frac{\log{m}}{m}}\right)$ 
w.h.p..
\end{theorem}

In this paper, we will use the transferability property of GNNs, together with a novel graph sampling algorithm, to train GNNs on small graph subsamples and ensure they scale well to large graphs.\vspace{2pt}

\noindent \textbf{Expressivity of GNNs.} While GNNs achieve remarkable performance in many graph machine learning tasks, they have fundamental limitations associated with their expressive power \cite{xu2018how,chen2019equivalence}. In GSP problems specifically, the expressivity of a GNN is constrained by the expressivity of the graph convolution, which in turn is constrained by the rank of the shift $\bbS$ \cite{ruiz2024spectral}. This is demonstrated in the following proposition.

\begin{proposition}[Expressivity of graph convolution]
Let $G$ be an $n$-node symmetric graph with rank-$r$ GSO $\mathbf{S}$, $r < n$, and $\bbx \in \reals^n$ an arbitrary graph signal. Consider the graph convolution $\hat{\bby} = \sum_{k=0}^{K-1} h_k \mathbf{S}^k \bbx$. Let $\ccalY \subset \reals^n$ be the subspace of signals that can be expressed as $\bby = \hat{\bby}$ for some $h_0, \ldots, h_{K-1}$. Then, $\dim(\ccalY) \leq r + 1$.
\end{proposition}

\begin{proof}
From the definition of graph convolution, we have
$$\ccalY = \text{span}(\{\bbS^k \bbx : k=0, \ldots, K\}).$$
Since $\bbS$ is symmetric, there exists some $\bbV, \mathbf{\Lambda} \in \mathbb{R}^{n \times n}$ satisfying $\bbV\bbV^\top = \bbI_n$ and $\mathbf{\Lambda}$ diagonal, such that $\bbS = \bbV\mathbf{\Lambda}\bbV^\top$. This implies the decomposition $\bbS^k = \bbV\mathbf{\Lambda}^k\bbV^\top$. As a result, $\bbS^k$ shares not only the same rank, but also the same spectral coordinates regardless of scaling. So for any given $\bbx\in \reals^n$, rank($\{\bbS^k\bbx : k=1, \ldots, K\}$) $\leq$ rank$(\bbS) = r$. Therefore, we have $\mathrm{dim}(\ccalY) \leq r+1$ with $\bbx$ as an additional coordinate candidate.
\end{proof}

In other words, the space of signals that can be represented with a graph convolution shrinks with the rank of the graph shift. Rank preservation is hence an important consideration when sampling subgraphs for training GNNs. To avoid the combinatorial complexity of evaluating rank directly, we will use the trace as a surrogate measure. This choice is analogous to objectives used in A-optimal experimental design \cite{pukelsheim2006optimal} and in convex relaxations of rank constraints for positive semidefinite matrices \cite{recht2010guaranteed,fazel2002matrix}.

\subsection{Problem Statement}\label{subsec:prob_statement}

When the graph structure is fully observed, preserving the Laplacian trace provides a natural surrogate for rank preservation. Since
\[
\mathrm{tr}(\bbL_c) = \sum_{i \in {V}} d_i,
\]
where $d_i$ is the degree of node $i$, this objective suggests degree-based sampling criteria, such as prioritizing nodes with large degree. However, in the absence of reliable structural information, preserving the Laplacian trace through purely graph-based heuristics is insufficient. The structural quantities required to evaluate or optimize trace may be only partially observed or unreliable in practice, and sampling decisions based solely on the observed graph may fail to preserve the spectral structure of interest.

We therefore use node features as an additional source of information to guide sampling designs. In many practical settings, feature statistics reflect how strongly individual nodes interact with the underlying graph structure. This can be exploited to guide sampling without requiring explicit spectral decompositions or detailed access to the graph structure. We exploit two complementary ways in which this interaction is encoded in the features, which we refer to as coarse and fine alignment.


Coarse alignment applies to feature-homophilic settings (Section~\ref{sec:coarse_alignment}), or equivalently to cases where node features have low Dirichlet energy. In this regime, features vary smoothly across the graph, and their energy is dominated by global structure rather than by local interactions. Consequently, nodes with large individual feature energy tend to contribute less to the relational structure captured by the Laplacian. Coarse alignment exploits this separation by prioritizing the removal of such nodes, preserving the structural information that is most relevant for graph convolution while avoiding costly spectral computations.

Fine alignment applies to settings in which the centered node features can be modeled as samples from a stationary graph signal. Under this assumption (see \cite{perraudin2017stationary,marques2017stationary} and Section~\ref{sec:fine_alignment}), the feature covariance admits a spectral characterization in terms of the graph Laplacian. When the corresponding filter is known, or when its frequency response is monotone, this structure enables us to infer the alignment between the diagonal of the Laplacian and the node-wise feature variances with high probability. Nodes can then be retained or removed according to this inferred alignment, yielding a more targeted sampling strategy.

\section{Coarse Alignment via Feature Heterophily}
\label{sec:coarse_alignment}

We begin by introducing the notion of \emph{feature heterophily}, which forms the basis of coarse alignment. To ensure the definition is comparable across graphs of different sizes and feature scales, we normalize the feature matrix $\bbX \in \reals^{n\times d}$ as
\begin{equation}\label{eqn:norm_features}
    \bbX \;=\; \frac{\bbX}{\max_{i,j}|X_{ij}|\sqrt{d}} \, .
\end{equation}

\begin{definition}[Feature heterophily]\label{def:feature_heterophily}
Let $G$ be an undirected graph with Laplacian $\bbL_c$ and normalized feature matrix $\bbX$ defined in~\eqref{eqn:norm_features}. The \emph{feature heterophily} of $G$ is
\begin{equation}\label{eqn:feature_heterophily}
    h_G \;=\; \frac{1}{n}\,\mathrm{tr}\!\left(\bbL_c\,\bbX\bbX^\top\right).
\end{equation}
\end{definition}

Since both $\bbL_c$ and $\bbX\bbX^\top$ are positive semidefinite, $h_G \ge 0$. Small values of $h_G$ indicate that the feature energy concentrates on low-frequency eigenspaces of $\bbL_c$, which encode global graph structure (e.g., communities). Conversely, larger values of $h_G$ indicate greater alignment with high-frequency components, which are typically associated with local or noisy variations. We will loosely call $G$ \emph{feature-homophilic} when $h_G$ is close to $0$.

The following result establishes a lower bound on $\mathrm{tr}(\bbL_c)$ in terms of $h_G$.

\begin{proposition}\label{prop:bound_trL}
For any undirected graph $G$ with normalized feature matrix $\bbX$ in~\eqref{eqn:norm_features},
\begin{equation}\label{eqn:bound_tr_L}
    \mathrm{tr}(\bbL_c) \;\ge\; \frac{n\, h_G}{\mathrm{tr}(\bbX\bbX^\top)} \, .
\end{equation}
\end{proposition}

\begin{proof}
Let $\bbC := \bbX\bbX^\top \succeq \mathbf{0}$. Since $\lambda_{\max}(\bbC)\le \mathrm{tr}(\bbC)$, we have
\[
\bbC \preceq \mathrm{tr}(\bbC)\,\mathbf{I}_n.
\]
Taking the trace after multiplying by $\bbL_c\succeq \mathbf{0}$ and using cyclic invariance yields
\[
\mathrm{tr}(\bbL_c\bbC) \le \mathrm{tr}\!\left(\bbL_c\,\mathrm{tr}(\bbC)\mathbf{I}_n\right)
= \mathrm{tr}(\bbC)\,\mathrm{tr}(\bbL_c).
\]
Recalling that $n h_G = \mathrm{tr}(\bbL_c\bbX\bbX^\top)=\mathrm{tr}(\bbL_c\bbC)$ gives
\[
\mathrm{tr}(\bbL_c)\ge \frac{n h_G}{\mathrm{tr}(\bbX\bbX^\top)},
\]
which is \eqref{eqn:bound_tr_L}.
\end{proof}

\subsection{Sampling Algorithm}

Despite its simplicity, Proposition~\ref{prop:bound_trL} provides insight for designing a sampling method on feature-homophilic graphs. For $h_G \approx 0$, the bound in~\eqref{eqn:bound_tr_L} becomes tighter as $\mathrm{tr}(\bbX\bbX^\top)$ decreases. This motivates removing nodes whose contributions to $\mathrm{tr}(\bbX\bbX^\top)=\sum_{i=1}^n \|X_{i,:}\|_2^2$ are largest, i.e., deleting nodes according to the diagonal entries of $\bbX\bbX^\top$ sorted in decreasing order. Equivalently, we retain nodes with the smallest diagonal scores; see Algorithm~\ref{alg:sampling}.

In the feature-homophilic regime, dominant directions of $\bbX\bbX^\top$ tend to align with low-frequency eigenspaces of $\bbL_c$, including the constant eigenvector associated with $\lambda_1(\bbL_c)=0$. Algorithm~\ref{alg:sampling} therefore tends to preserve Laplacian rank whenever possible. Recall that the multiplicity of the zero eigenvalue of $\bbL_c$ equals the number of connected components of $G$. Removing an isolated node (a single-node component) decreases both $n$ and the number of components by one, leaving $\mathrm{rank}(\bbL_c)=n-k$ unchanged.

Finally, Algorithm~\ref{alg:sampling} is related to \emph{leverage-score sampling} in randomized numerical linear algebra \cite{woodruff2014sketching,mahoney2011randnla}. The scores $s$ can be interpreted as leverage scores weighted by squared singular values, except that we remove (rather than retain) nodes with the highest scores. To see the connection, recall that leverage scores are the diagonal entries of the hat matrix $\bbH = \bbX(\bbX^\top \bbX)^{+} \bbX^\top$. Let the SVD of $\bbX$ be $\bbX = \bbU \bbSigma \bbV^\top$ with $\bbU = [\bbU_r \;\; \bbU_\perp]$ and $\bbSigma = \mathrm{diag}(\bbSigma_r, 0)$. Then $\bbH = \bbU_r \bbU_r^\top$, while $\bbX \bbX^\top = \bbU \bbSigma^2 \bbU^\top$, so $\mathrm{diag}(\bbX \bbX^\top)$ corresponds to leverage scores weighted by the squared singular values.

\begin{algorithm}[t]
\caption{Node Sampling for Feature-Homophilic Graphs}
\label{alg:sampling}
\begin{algorithmic}[1]
    \Require Graph $G=(V,E)$ with $|V|=n$, feature matrix $\bbX \in \reals^{n \times d}$, sampling ratio $\gamma \in [0,1]$
    \State $n_d \gets \lfloor (1-\gamma)n \rfloor$ \Comment{deletion budget}
    \State $\boldsymbol{s} \gets \mathrm{diag}(\bbX\bbX^\top)$ \Comment{node scores}
    \State $\text{idx} \gets \mathrm{argsort}(\boldsymbol{s})[0 : n-n_d]$ \Comment{retain smallest scores}
    \State $\tilde{V} \gets \text{idx}$
    \State $\tilde{E} \gets \{(u,v)\in E:\; u\in\tilde{V},\, v\in\tilde{V}\}$
    \State $\tilde{\bbX} \gets X_{\text{idx},:}$
    \State \Return $\tilde{G}(\tilde{V},\tilde{E})$, $\tilde{\bbX}$
\end{algorithmic}
\end{algorithm}

\subsection{Computational Complexity}

When $d \ll n$, as is common in large-scale graphs, Algorithm~\ref{alg:sampling} has lower complexity than rank-preserving spectral methods. The following proposition formalizes this claim.

\begin{proposition}\label{prop:complexity}
Let $G=(V,E)$ have $|V|=n$ and $|E|=m$, and let $\bbX\in\reals^{n\times d}$ denote the node feature matrix. For any sampling ratio $\gamma$, the overall complexity of Algorithm~\ref{alg:sampling} is $\mathcal{O}(dn + n\log n)$, dominated by computing and sorting $\boldsymbol{s}=\mathrm{diag}(\bbX\bbX^\top)$. If, in addition, one computes $h_G$ via sparse multiplications using $\bbL_c=\bbD-A$, then evaluating $h_G=\frac{1}{n}\mathrm{tr}(\bbL_c\bbX\bbX^\top)$ costs $\mathcal{O}(dm)$.
\end{proposition}

\begin{proof}
Computing the scores $\boldsymbol{s}=\mathrm{diag}(\bbX\bbX^\top)$ requires $\mathcal{O}(dn)$ operations since $s_i=\|X_{i,:}\|_2^2$. Sorting $\boldsymbol{s}$ costs $\mathcal{O}(n\log n)$. Constructing the induced subgraph by filtering edges costs $\mathcal{O}(m)$ if edges are scanned once and membership in $\tilde{V}$ is checked in $\mathcal{O}(1)$ time (e.g., via a Boolean mask).

To compute $h_G$, note that $\bbL_c\bbX=(\bbD-\bbA)\bbX$ can be formed as $\bbD \bbX$ (cost $\mathcal{O}(dn)$) minus $\bbA \bbX$ (cost $\mathcal{O}(dm)$ for sparse $\bbA$). The trace $\mathrm{tr}(\bbL_c\bbX\bbX^\top)=\mathrm{tr}(\bbX^\top\bbL_c\bbX)$ then costs an additional $\mathcal{O}(dn)$. Thus the total cost is $\mathcal{O}(dm)$.
\end{proof}

On sparse graphs ($m\ll n^2$) with moderate feature dimension $d$, the proposed method is substantially cheaper than direct maximization of $\mathrm{tr}(\bbL_c)$ under sequential deletions, which requires updating degrees and incurs at least $\mathcal{O}((1-\gamma)n^2)$ operations in dense representations. It is also significantly cheaper than spectral approaches such as \cite{chen2015discrete,anis2016efficient}, which rely on sequential greedy heuristics. In contrast, Algorithm~\ref{alg:sampling} admits parallel implementation, further improving scalability.
\section{Fine Alignment via Stationarity}
\label{sec:fine_alignment}
While coarse alignment suffices to preserve global spectral mass, it does not distinguish which nodes contribute most to that mass. Fine alignment addresses this limitation by exploiting the notion of graph signal stationarity \cite{perraudin2017stationary,marques2017stationary} to recover node-wise structural ordering from feature statistics.

From now on, we assume that the feature dimension $d$ and the number of nodes $n$ satisfy the scaling $d^{1/\alpha} \leq n \leq d^\alpha$ for some fixed $\alpha > 0$. This assumption is standard in high-dimensional statistics and encompasses the linear regime ($n \propto d$) commonly studied in random matrix theory \cite{bai2010spectral,vershynin2018hdp}, as well as more general polynomial scalings. It rules out only extreme settings in which one dimension grows exponentially faster than the other.

Let $\bbL := \bbL_c/n$ denote the rescaled Laplacian of an $n$-node undirected graph. Under the stationarity assumption, node features are generated through convolutional filtering of latent excitation signals on the graph, inducing a tight relationship between the feature covariance and the spectrum of $\bbL$.

\begin{definition}[Stationary graph signal model \cite{perraudin2017stationary,marques2017stationary}]\label{def:stationary_model}
Let $\mathbf{x}_1, \ldots, \mathbf{x}_n \sim \mathcal{N}(0, \bbI_d)$ be independent and identically distributed random vectors in $\mathbb{R}^d$, grouped in the matrix $\bbX_0$ as
\begin{equation}
\mathbf{X}_0 =\frac{1}{\sqrt{d}}\begin{bmatrix} \mathbf{x}_1^\top \\
\vdots \\
\mathbf{x}_n^\top
\end{bmatrix}
\in \mathbb{R}^{n \times d}.
\end{equation}
I.e., each row of $\mathbf{X}_0$ corresponds to a feature vector associated with a graph node. The normalization by $\sqrt{d}$ ensures  $\mathbb{E}[\mathbf{X}_0\mathbf{X}_0^\top] = \mathbf{I}$.
Fix a measurable function (filter) $h$. The function $h$ defines the stationary graph signal 
\begin{equation}
\mathbf{X} = h(\bbL) \mathbf{X}_0,
\label{eq:model1}
\end{equation}
where $h(\bbL)$ denotes application of $h$ to the rescaled Laplacian $\bbL$ in the spectral domain (cf. \eqref{eqn:spectral_filter}).
\end{definition}

This model corresponds to applying a spectral graph filter---such as the graph convolution in \eqref{eqn:simple_conv}---to an isotropic input signal. By construction, the rows of $\mathbf{X}_0$ are independent with $\mathbb{E}[\mathbf{X}_0 \mathbf{X}_0^\top] = \mathbf{I}$, so the resulting feature matrix $\mathbf{X}$ carries graph structural information through the graph filter $h(\bbL)$.

\subsection{Operator Alignment via Weak Commutativity}
\label{subsec:weak_commutativity}
It can be shown that in the stationary signal model, the feature covariance inherits the spectral structure of the Laplacian operator via commutativity. In the following analysis, we consider the general case where the graph—and thus the Laplacian $\mathbf{L}$—may be random from some random graph model. Since $h(\mathbf{L})$ is a function of a normal operator, $\mathbb{E}[\mathbf{X}\mathbf{X}^\top]$ and $\mathbf{L}$ admit a common eigendecomposition. This is formalized in the following proposition.

\begin{proposition}[Commutativity in expectation]\label{prop:expected_alignment}
Under the stationary graph signal model (Definition \ref{def:stationary_model}), assuming the initial signal $\mathbf{X}_0$ is sampled independently from the (possibly random) graph Laplacian $\mathbf{L}$, the expected covariance matrix and the Laplacian commute:
\begin{equation}
    \mathbb{E}[\mathbf{X}\mathbf{X}^\top \mathbf{L}] = \mathbb{E}[\mathbf{L} \mathbf{X}\mathbf{X}^\top].
\end{equation}
Or, equivalently, $\mathbb{E}[ [\mathbf{X}\mathbf{X}^\top, \mathbf{L}] ] = \mathbf{0}$, where $[\mathbf{A}, \mathbf{B}] = \mathbf{AB} - \mathbf{BA}$ denotes the commutator operator.
\end{proposition}

\begin{proof}
By the law of total expectation and the independence of $\mathbf{X}_0$ and $\mathbf{L}$, we have:
\begin{align*}
\mathbb{E}[\mathbf{X}\mathbf{X}^\top \mathbf{L}] &= \mathbb{E}\big[ \mathbb{E}[\mathbf{X}\mathbf{X}^\top \mid \mathbf{L}] \mathbf{L} \big] \\
&= \mathbb{E}[ h(\mathbf{L}) \mathbb{E}[\mathbf{X}_0 \mathbf{X}_0^\top] h(\mathbf{L}) \mathbf{L} ] \\
&= \mathbb{E}[ h^2(\mathbf{L}) \mathbf{L} ].
\end{align*}
Since any spectral function $h^2(\mathbf{L})$ commutes with $\mathbf{L}$ (i.e., $h^2(\mathbf{L})\mathbf{L} = \mathbf{L}h^2(\mathbf{L})$), it follows that:
\begin{align*}
\mathbb{E}[ h^2(\mathbf{L}) \mathbf{L} ] &= \mathbb{E}[ \mathbf{L} h^2(\mathbf{L}) ] \\
&= \mathbb{E}\big[ \mathbf{L} \mathbb{E}[\mathbf{X}\mathbf{X}^\top \mid \mathbf{L}] \big] \\
&= \mathbb{E}[\mathbf{L} \mathbf{X}\mathbf{X}^\top].
\end{align*}
This concludes the proof.
\end{proof}
This expectation-level commutativity suggests a shared spectral structure, but does not yet imply that a single realization of $\bbX\bbX^\top$ reflects the graph Laplacian $\bbL$. Bridging this gap requires finite-sample control.
I.e., what we need is for the covariance matrix $\bbX\bbX^\top$ and the Laplacian $\bbL$ to commute w.h.p.

To prove that this indeed holds, we first require the following lemma 
showing that the entries of the $\bbX\bbX^\top$ are close to the entries of $h^2(\bbL)$. This lemma leverages the fact that $\bbX_0$ is isotropic together with Bernstein's inequality \cite{vershynin2018hdp,wainwright2019high}.

\begin{lemma}\label{lem:weak approximation}
Under model \eqref{eq:model1}, there exist constants $c,C_h>0$ such that, for any deterministic unit vectors $\mathbf{u},\mathbf{v}\in\mathbb{R}^n$, we have
\[
 \mathbb{P}\left(|\mathbf{v}^\top(\mathbf{X}\mathbf{X}^\top-h^2(\mathbf{L}))\mathbf{u}|\geq\frac{C_h\log(d)}{\sqrt{d}}\right)\leq 2d^{-c\log(d)}.
\]
\end{lemma}

\begin{proof}
By the definition of $\mathbf{X}$, we have:
\begin{align*}
\mathbf{X}\mathbf{X}^\top &= h(\mathbf{L})\mathbf{X}_0\mathbf{X}_0^\top h(\mathbf{L}) \\
 &= h^2(\mathbf{L}) + h(\mathbf{L})(\mathbf{X}_0\mathbf{X}_0^\top - \mathbf{I})h(\mathbf{L}).
\end{align*}
For any deterministic unit vectors $\mathbf{u}, \mathbf{v} \in \mathbb{R}^n$, we define the term $\mathcal{S}$ as:
\begin{align*}
 \mathbf{v}^\top(\mathbf{X}\mathbf{X}^\top - h^2(\mathbf{L}))\mathbf{u} &= \mathbf{v}^\top h(\mathbf{L})(\mathbf{X}_0\mathbf{X}_0^\top - \mathbf{I})h(\mathbf{L})\mathbf{u} =: \mathcal{S}.
\end{align*}
Let $\mathbf{v}' = h(\mathbf{L})\mathbf{v}$ and $\mathbf{u}' = h(\mathbf{L})\mathbf{u}$. We can rewrite $\mathcal{S}$ as:
\begin{align*}
\mathcal{S} &= \mathbf{v}'^\top(\mathbf{X}_0\mathbf{X}_0^\top - \mathbf{I})\mathbf{u}' \\
&= \frac{1}{d} \sum_{\alpha=1}^d \left[ (\sqrt{d}\mathbf{e}_\alpha^\top \mathbf{X}_0^\top \mathbf{v}')(\sqrt{d}\mathbf{e}_\alpha^\top \mathbf{X}_0^\top \mathbf{u}') - \mathbf{v}'^\top \mathbf{u}' \right],
\end{align*}
where $\{\mathbf{e}_\alpha\}_{\alpha=1}^d$ denotes the canonical basis vectors in $\mathbb{R}^d$. 

Since $\mathbf{X}_0$ has i.i.d. $N(0, 1/d)$ entries, the terms $(\sqrt{d}\mathbf{e}_\alpha^\top \mathbf{X}_0^\top \mathbf{v}')$ and $(\sqrt{d}\mathbf{e}_\alpha^\top \mathbf{X}_0^\top \mathbf{u}')$ are independent Gaussian random variables. Their product is sub-exponential with a sub-exponential norm bounded by $\|\mathbf{v}'\|_2 \|\mathbf{u}'\|_2 \leq C_h$, where $C_h$ depends only on the spectral norm of $h(\mathbf{L})$. 

Applying Bernstein's inequality for sub-exponential random variables \cite[Theorem 2.8.1]{vershynin2018high}, we obtain for any $t > 0$:
\[
 \mathbb{P}(|\mathcal{S}| \geq t/d) \leq 2 \exp\left( -c \min\left\{ \frac{t^2}{C_h^2 d}, \frac{t}{C_h} \right\} \right).
\]
Setting $t = C_h \log(d) \sqrt{d}$ yields the desired bound.
\end{proof}

Using Lemma~\ref{lem:weak approximation}, we prove commutativity at the sample level by showing that quadratic forms of the empirical covariance $\bbX\bbX^\top$ concentrate uniformly around those of $h^2(\bbL)$. This provides the finite-sample control needed to move from expectation to realization.

\begin{proposition}[Weak commutativity]\label{prop:weak_commutativity}
Under model \eqref{eq:model1}, the Laplacian matrix $\mathbf{L}$ and the covariance matrix $\mathbf{X}\mathbf{X}^\top$ commute weakly: there exist constants $c,C_h>0$ such that, for any deterministic unit vectors $\mathbf{u},\mathbf{v}\in\mathbb{R}^n$, we have
\begin{equation}
\begin{aligned}
    \mathbb{P}\Bigg( \Big| \mathbf{v}^\top(\mathbf{L} \mathbf{X}\mathbf{X}^\top &- \mathbf{X}\mathbf{X}^\top\mathbf{L})\mathbf{u} \Big| \\
    &\geq \frac{C_h\log(d)}{\sqrt{d}} \Bigg) \leq 4d^{-c\log(d)}.
\end{aligned}
\end{equation}
\end{proposition}

\begin{proof}
Since $h^2(\mathbf{L})$ commutes with $\mathbf{L}$, we have for any deterministic unit vectors $\mathbf{u}, \mathbf{v} \in \mathbb{R}^n$:
\begin{align*}
 \mathbf{v}^\top(\mathbf{L} \mathbf{X}\mathbf{X}^\top - \mathbf{X}\mathbf{X}^\top\mathbf{L})\mathbf{u} &= \mathbf{v}^\top \mathbf{L}(\mathbf{X}\mathbf{X}^\top - h^2(\mathbf{L}))\mathbf{u} \\
 &\quad - \mathbf{v}^\top (\mathbf{X}\mathbf{X}^\top - h^2(\mathbf{L}))\mathbf{L}\mathbf{u}.
\end{align*}
We define the two resulting terms as:
\begin{align*}
 \mathcal{S}_1 &:= \mathbf{v}^\top \mathbf{L} h(\mathbf{L})(\mathbf{X}_0\mathbf{X}_0^\top - \mathbf{I})h(\mathbf{L})\mathbf{u}, \\
 \mathcal{S}_2 &:= \mathbf{v}^\top h(\mathbf{L})(\mathbf{X}_0\mathbf{X}_0^\top - \mathbf{I})h(\mathbf{L})\mathbf{L}\mathbf{u}.
\end{align*}
The result follows from applying Lemma \ref{lem:weak approximation} to $\mathcal{S}_1$ and $\mathcal{S}_2$ individually. Specifically, by defining $\tilde{\mathbf{v}} = \mathbf{L}\mathbf{v}/\|\mathbf{L}\mathbf{v}\|_2$ and $\tilde{\mathbf{u}} = \mathbf{L}\mathbf{u}/\|\mathbf{L}\mathbf{u}\|_2$, both terms satisfy the concentration bound. A union bound over $\mathcal{S}_1$ and $\mathcal{S}_2$ yields the final probability.
\end{proof}


\subsection{Nodewise Alignment and Sampling}

The key implication of the commutativity result in Proposition \ref{prop:weak_commutativity} is that, beyond cospectrality, it also induces alignment at the level of the diagonal entries of $\bbX\bbX^\top$ and $\bbL$. This diagonal alignment reveals node-wise structural information encoded in $\bbX\bbX^\top$, which we exploit as a form of \emph{fine alignment} to guide sampling.

\begin{theorem}[Fine alignment]\label{thm:diag_concentration}
Under model \eqref{eq:model1}, for any $\epsilon > 0$, there exist constants $c,C_h>0$ such that
\begin{align*}
 &\mathbb{P} \left( \max_{1 \leq i \leq n} \left| \mathbf{e}_i^\top \mathbf{X}\mathbf{X}^\top \mathbf{e}_i - \mathbf{e}_i^\top h^2(\bbL) \mathbf{e}_i \right| \geq \frac{C_h \log}{\sqrt{d}} \right)\\
&\leq 2d^{-c\log d}.
\end{align*}
\end{theorem}
\begin{proof}
 Follows directly from Lemma \ref{lem:weak approximation}.
\end{proof}

In words, under the stationary graph signal model, the diagonal entries of the empirical feature covariance $\mathbf{X}\mathbf{X}^\top$ concentrate uniformly around the diagonal of $h^2(\bbL)$ w.h.p. As a consequence, node-wise feature variances provide accurate estimates of the corresponding diagonal entries of a function of the graph Laplacian, up to fluctuations that vanish as the feature dimension grows, establishing a direct link between observable feature statistics and structural quantities associated with the graph operator.

When the filter $h$ is known explicitly, Theorem~\ref{thm:diag_concentration} can be used to infer a degree-based sampling rule from features alone. Indeed, the theorem implies that the observable quantities $\{(\mathbf{X}\mathbf{X}^\top)_{ii}\}_{i=1}^n$ provide accurate estimates of the diagonal entries of $h^2(\bbL)$. For common choices of $h$, these diagonal entries admit a simple dependence on local connectivity and are monotone in the node degree; see Example~\ref{ex:identity_filter} and Example~\ref{ex:positive_poly_filter}. Consequently, sorting nodes by $(\mathbf{X}\mathbf{X}^\top)_{ii}$ recovers the degree ordering w.h.p., enabling a trace-preserving sampling strategy---i.e., retaining nodes with large inferred degree---without direct access to the graph structure.
The following examples show how this monotonicity may arise in practice, and how it can be used to inform feature-driven sampling.

\begin{example}[Identity filter]\label{ex:identity_filter}
Suppose that $h(x) = x$, so that $h^2(\bbL) = \bbL^2$. Then, for all $1 \leq i \leq n$:
\[
    [\bbL^2]_{ii} = \mathbf{e}_i^\top \bbL^2 \mathbf{e}_i = \frac{1}{n^2} \left( d_i^2 + \sum_{(i,j) \in E} 1 \right) = \frac{1}{n^2} (d_i^2 + d_i),
\]
where $d_i$ is the degree of node $i$. This expression shows that the diagonal of $\bbL^2$ is strictly monotonic with respect to node degree. Recall that $\mathrm{tr}(\bbL) = 2|E|/n$. Therefore, a sampling procedure keeping nodes with large contributions to $\mathrm{tr}(\mathbf{X}\mathbf{X}^\top)$ is equivalent to preserving nodes with high local connectivity, which naturally preserves the Laplacian trace.
\end{example}

\begin{example}[General positive polynomial filter]\label{ex:positive_poly_filter}
Consider a general filter $h(x) = \sum_{k=0}^{d} a_k x^k$ where all coefficients $a_k$ are positive. Then we have:
\[
    h^2(\bbL) = \sum_{m=0}^{2d} c_m \bbL^m, \quad \text{where } c_m = \sum_{k + \ell = m} a_k a_\ell
\]
with diagonal entry at index $i$ given by:
\[
    \left[ h^2(\bbL) \right]_{ii} = \sum_{m=0}^{2d} c_m \left[ \bbL^m \right]_{ii}.
\]
Since the coefficients $a_k$ are positive, the filter function $h$ and, consequently, the squared response $h^2$, are monotone increasing with the node degree. Similarly to Example~\ref{ex:identity_filter}, this suggests a feature-based sampling procedure wherein we keep nodes with largest diagonal entries in $\bbX\bbX^\top$.
Perhaps more interestingly, the term $[\bbL^m]_{ii}$ quantifies the weight of closed walks of length $m$ starting and ending at node $i$. Given that the coefficients $a_k$ are positive, the coefficients $c_m$ are also positive. As such, a larger diagonal value in $\mathbf{X}\mathbf{X}^\top$ implies that node $i$ is central to the graph's cyclic structure, aggregating positive weights from walks of all lengths up to $2d$.

\end{example}


These examples further illustrate that explicit knowledge of the filter $h$ is not required to recover a degree-based sampling rule. It is sufficient that $h$ 
be monotone,
so that ordering nodes by $[\mathbf{X}\mathbf{X}^\top]_{ii}$ induces the same (if monotone increasing), or the reverse (if monotone decreasing), ordering as their contributions to the Laplacian trace. In settings where this monotonicity is not known a priori, it could also be inferred from ``low-order approximations'' such as the feature heterophily from Definition~\ref{def:feature_heterophily}. In this sense, coarse alignment can be interpreted as a low-order approximation to the fine alignment framework developed here.

\begin{remark}
Under fine alignment, the complexity of sampling is the same as that of Algorithm~\ref{alg:sampling}, minus the heterophily computation ~\ref{prop:complexity}. Unlike spectral sampling methods that rely on eigendecompositions and incur cubic complexity in the number of nodes (\(\mathcal{O}(n^3)\)), the diagonal of \(\mathbf{X}\mathbf{X}^\top\) provides a computationally efficient quantity that can be evaluated in \(\mathcal{O}(nd^2)\). Thus, when appropriate alignment conditions such as the ones in~\ref{ex:identity_filter} and~\ref{ex:positive_poly_filter}  hold, sampling nodes according to these diagonal entries offers an efficient way to target regions of high spectral energy induced by an increasing filter \(h\), preserving the Laplacian trace as a proxy for rank while avoiding explicit spectral computations.
\end{remark}

\begin{figure*}[t]
\centering
\includegraphics[height=4cm,width=0.32\textwidth]{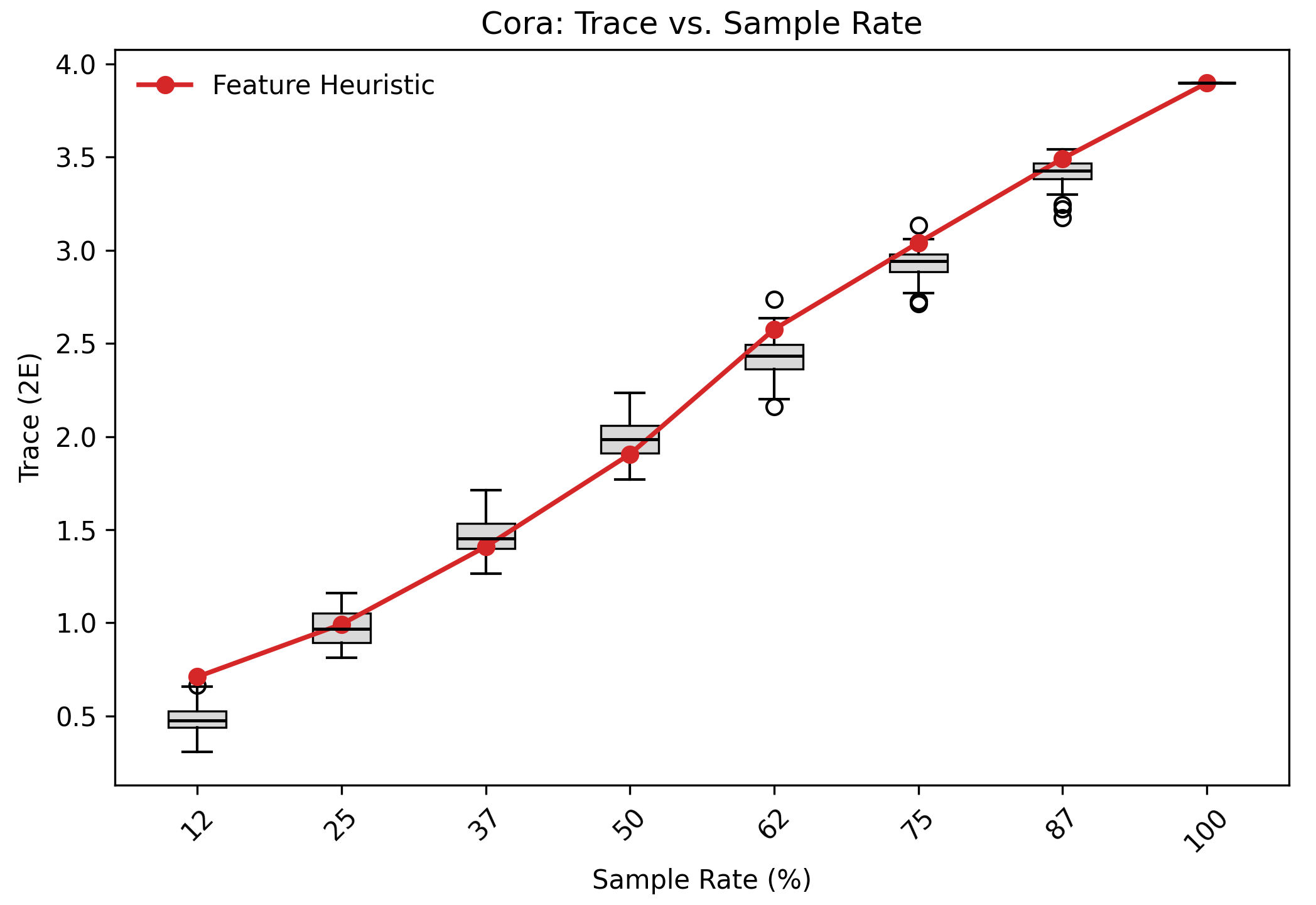}
\includegraphics[height=4cm,width=0.32\textwidth]{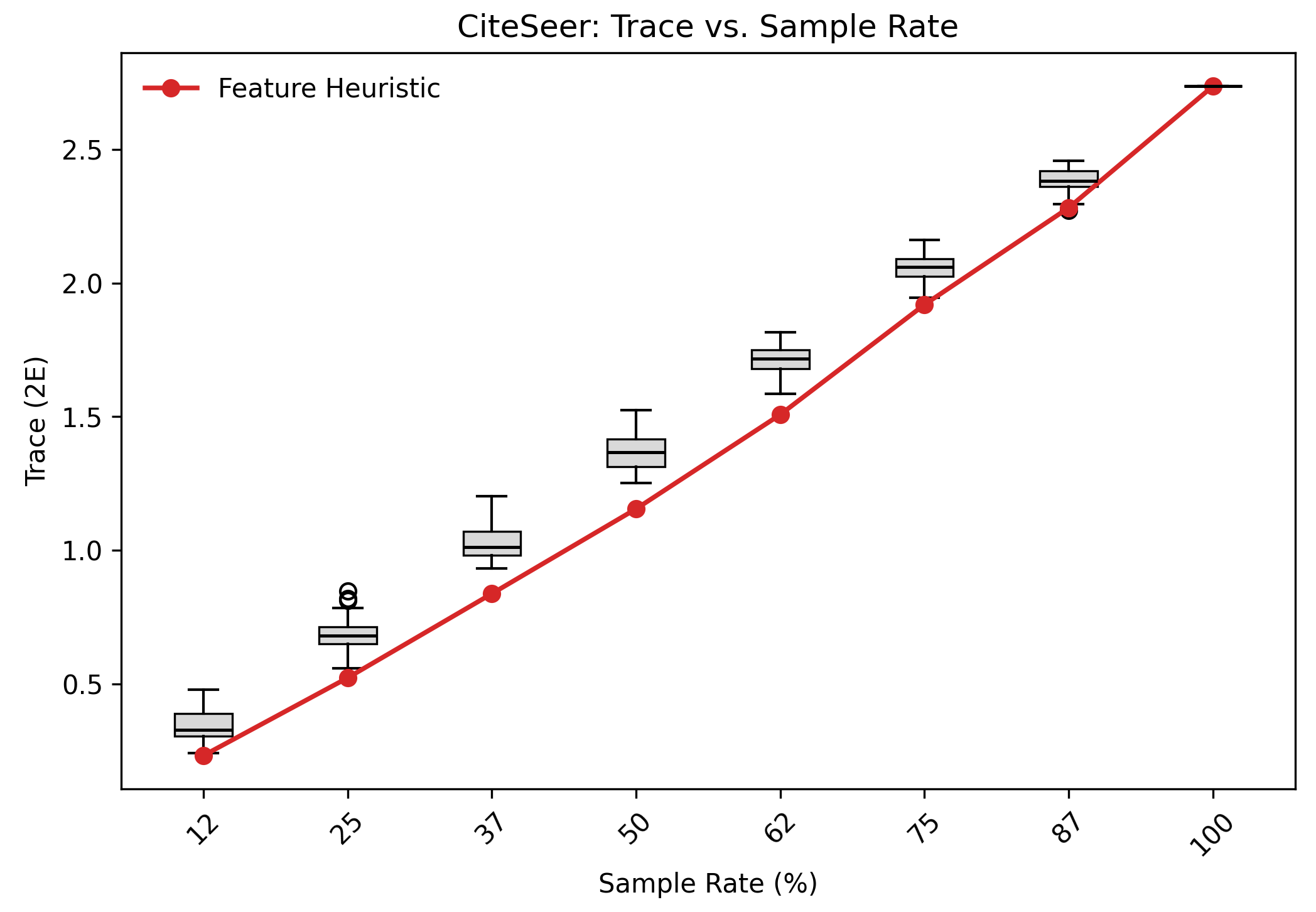}
\includegraphics[height=4cm,width=0.32\textwidth]{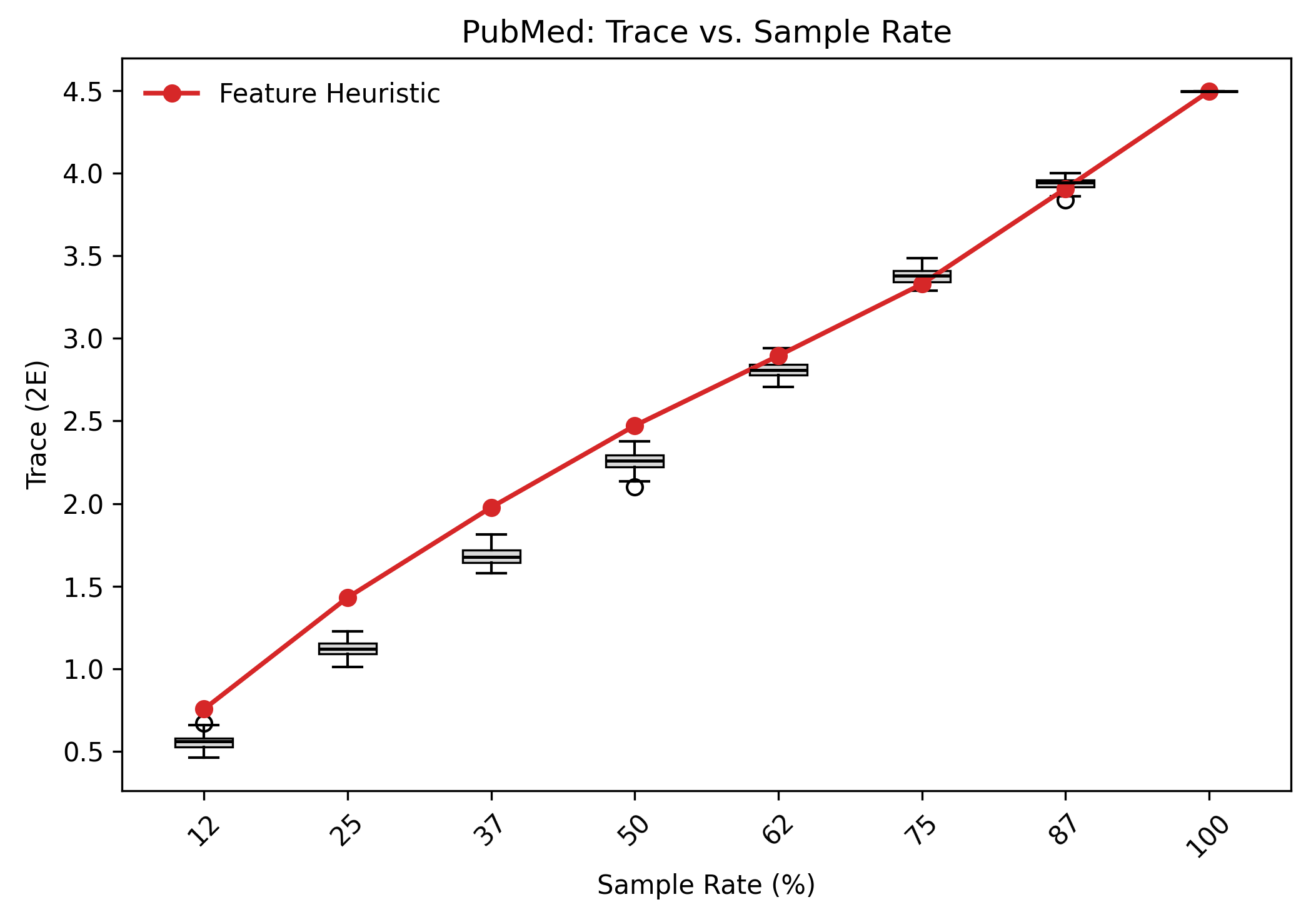}
\caption{Adjusted Laplacian trace versus graph subsampling rate. The adjusted trace is the subsampled graph Laplacian trace normalized by the number of sampled nodes. Boxplots indicate the median, first and third quartiles, minimum and maximum, and outliers of the trace, obtained from 50 rounds of random node subsampling; red dots are the trace of subgraphs generated using our sampling heuristic (Algorithm 1).}
\label{fig:trace}
\end{figure*}

\section{Case Study on Stochastic Block Model}\label{sec:sbm}

To further examine the implications of fine alignment in a more realistic setting, we consider a case study on stochastic block model (SBM) graphs. This random but controlled setting allows us to analyze the behavior of the proposed sampling strategy under two representative classes of graph filters: monotone increasing and monotone decreasing. 

In this section, we establish theoretical guarantees for the resulting sampling procedure by connecting three key ideas:
\begin{enumerate}
    \item \textbf{Spectral Approximation:} We show that the diagonal entries of $h^2(\mathbf{L})$ are locally determined by the normalized node degrees (Lemma \ref{lem:diagon of f squared}).
    \item \textbf{Degree-Ordering Preservation:} We prove that the ordering of the observable diagonal entries of $\mathbf{X}\mathbf{X}^\top$ correctly identifies high-degree or low-degree nodes depending on the filter's monotonicity (Theorem \ref{thm:large diagonal is large degree}).
    \item \textbf{Trace Optimality:} We conclude that feature-driven sampling based on these diagonal entries consistently outperforms uniform sampling in preserving the Laplacian trace (Corollary \ref{cor:trace_preservation}).
\end{enumerate}

An SBM with $r \geq 2$ communities is parameterized by a community size vector $\boldsymbol{n} = (n_1, \ldots, n_r) \in \mathbb{N}^r$ and a symmetric edge probability matrix $P \in [0,1]^{r \times r}$, where $P_{ij} = P_{ji}$. We denote by $\mathcal{G}(\boldsymbol{n}, P)$ the distribution over random graphs with $n = \sum_{i=1}^r n_i$ nodes, in which nodes are partitioned into $r$ disjoint communities of sizes $n_1, \ldots, n_r$, and edges between nodes in communities $i$ and $j$ are included independently with probability $P_{ij}$.
We further denote the expected adjacency and degree matrices $\overline{\mathbf{A}} = \mathbb{E}[A]$ and $\overline{\mathbf{D}} = \mathbb{E}[\mathbf{D}]$, respectively, and define the expected rescaled Laplacian as
\[
\overline{\bbL} := n^{-1} (\overline{\mathbf{D}} - \overline{A}).
\]

By Theorem~\ref{thm:diag_concentration}, we know that the diagonal entries of $\mathbf{X}\mathbf{X}^\top$ concentrate around those of $h^2(\bbL)$. So, our main goal is to analyze what the diagonal elements of $h^2(\bbL)$ represent. The following lemma shows that the $i$th diagonal entry of $h^2(\bbL)$ approximates the nodewise evaluation of $h^2$ at the normalized node degrees $d_i/n$.

\begin{figure*}[t]
\centering
\includegraphics[height=4cm,width=0.32\textwidth]{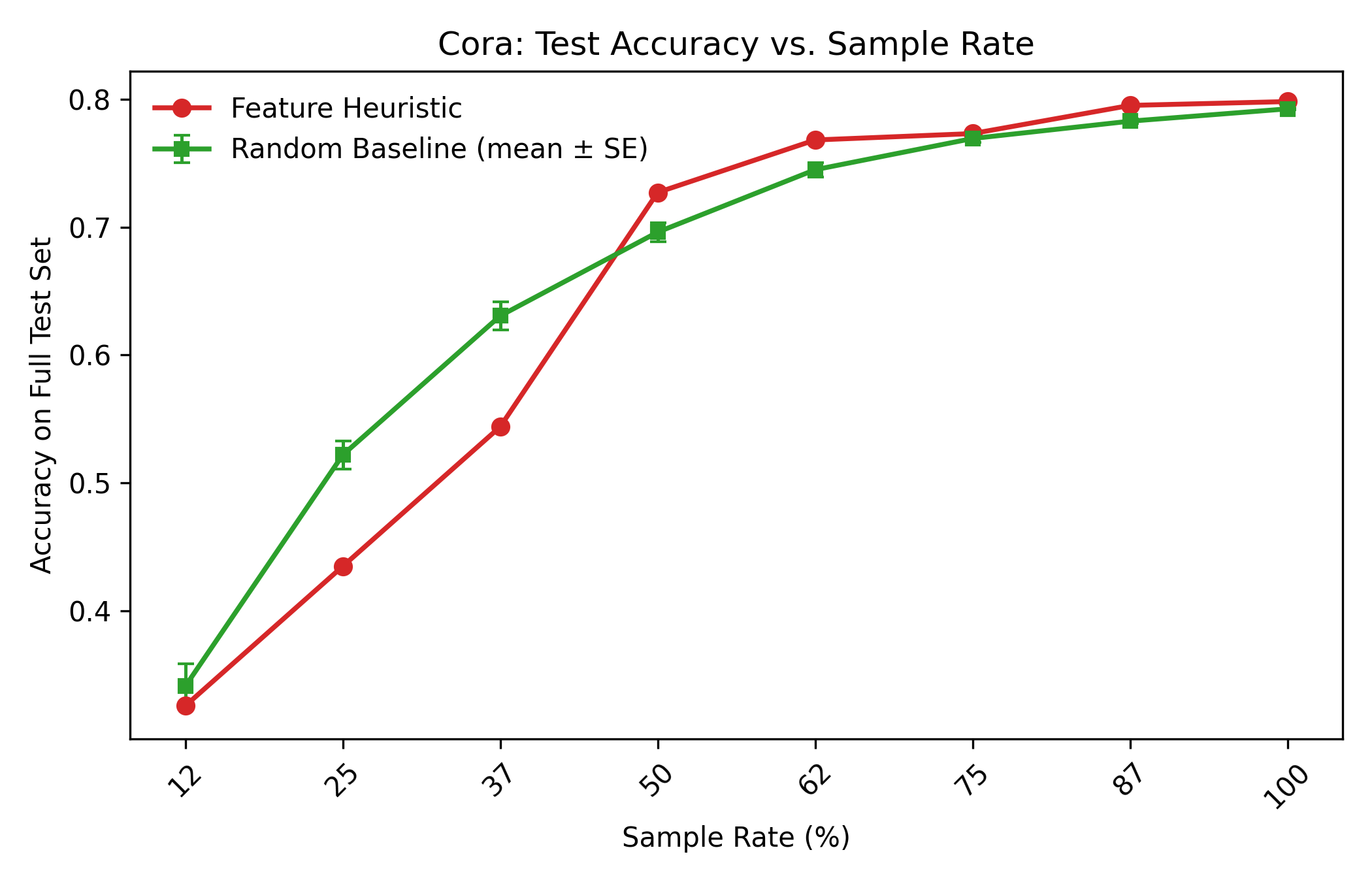}
\includegraphics[height=4cm,width=0.32\textwidth]{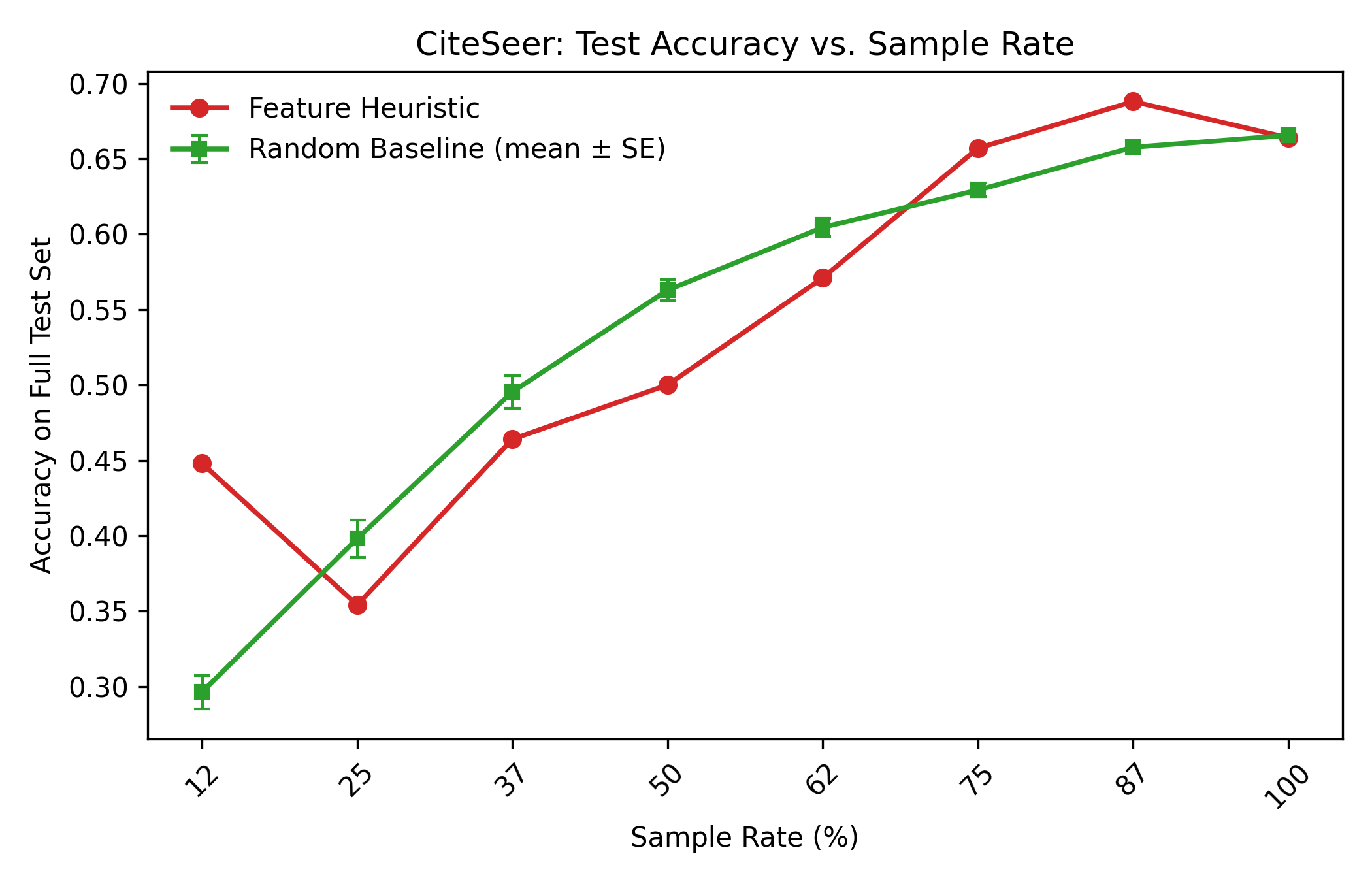}
\includegraphics[height=4cm,width=0.32\textwidth]{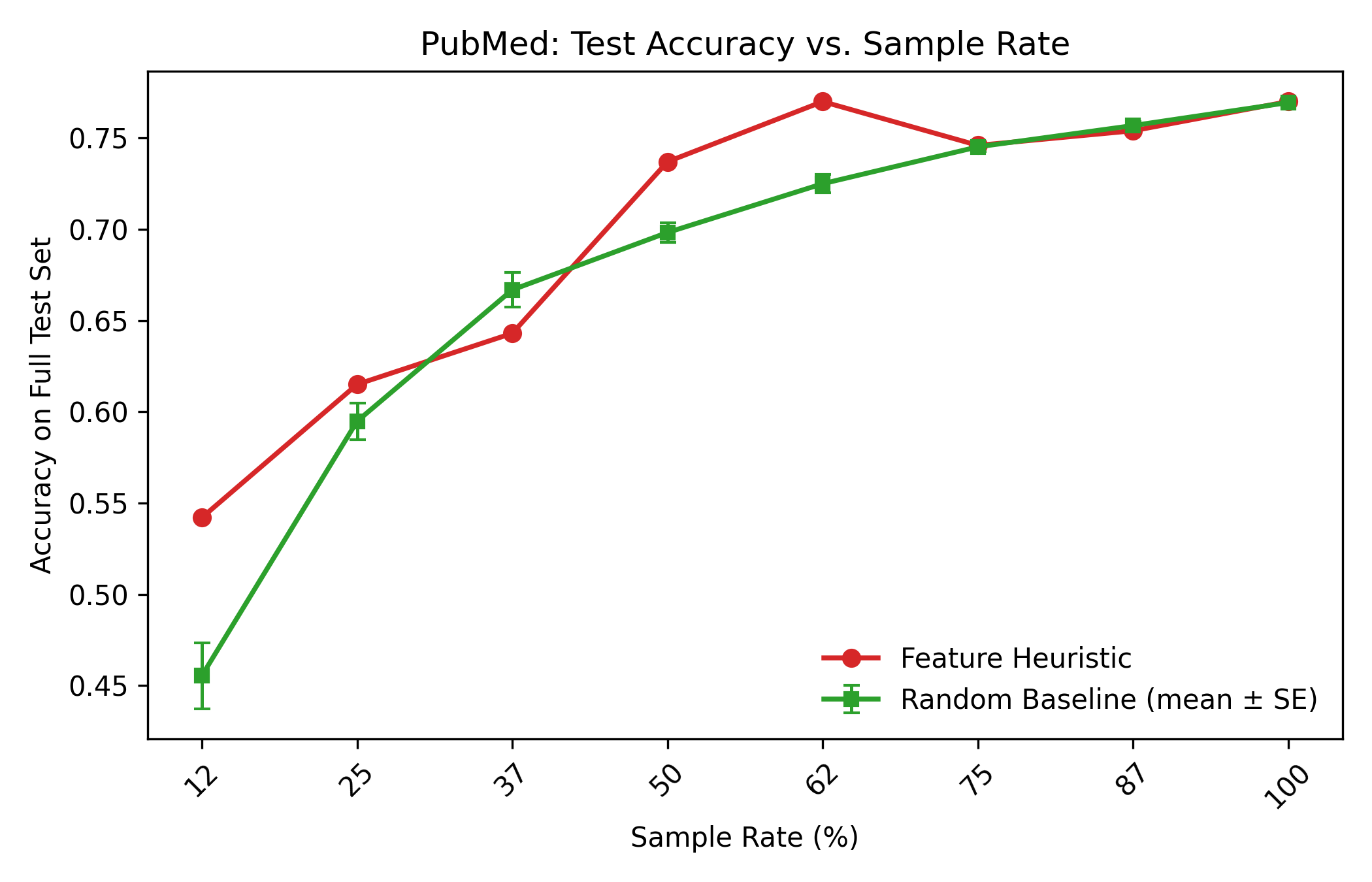}
\caption{Test accuracy achieved by GNN on full graph versus training graph subsampling rate. Error bars indicate the standard deviation of the accuracy realized by GNNs trained on 50 random node-induced subgraphs; red dots are the test accuracy of GNNs trained on subgraphs produced by our heuristic (Algorithm 1).}
\label{fig:GNN Acc}
\end{figure*}


\begin{lemma}[Diagonal approximation of $h^2(\bbL)$]
\label{lem:diagon of f squared}
Let $h: [0,2] \to \mathbb{R}$ be $\beta$-Hölder continuous with $0 < \beta \leq 1$. Suppose $G \sim \mathcal{G}(\boldsymbol{n}, \bbP)$ satisfies:
\[
\sum_{i > j} \overline{\bbA}_{ij}(1 - \overline{\bbA}_{ij}) \geq \delta_1 n^{1 + \delta_2}, \quad r \leq (\log n)^{\delta_1}, \quad \min\{n_s\} \geq n^{\delta_2}.
\]
Then for any $\tau > 0$, there exists $C > 0$ such that for all $i \in [n]$,
\[
\left| \left[h^2(\bbL)\right]_{ii} - h^2\left( \frac{d_i}{n} \right) \right| \leq C \cdot \max\left\{ n^{-\beta/2}, n^{-\delta_2} \right\}
\]
with probability at least $1 - Cn^{-\tau}$.
\end{lemma}
\begin{proof}
See Appendix~\ref{app:proof_spectral_concentration}.
\end{proof}

Hence, Lemma~\ref{lem:diagon of f squared} reduces the analysis of diagonal entries of $h^2(\bbL)$ to a function of the node degrees. This result, combined with the concentration of $(\mathbf{X}\mathbf{X}^\top)_{ii}$ around $\left[h^2(\bbL)\right]_{ii}$, allows us to establish degree-ordering guarantees under monotone increasing or decreasing filters. In particular, it quantifies the probability with which comparisons between diagonal entries of $\mathbf{X}\mathbf{X}^\top$ correctly reflect comparisons between node degrees under monotone increasing or decreasing filters, as stated in Theorem~\ref{thm:large diagonal is large degree}.

\begin{theorem}[Diagonal reflects degree]
\label{thm:large diagonal is large degree}
Consider the same assumptions of Lemma~\ref{lem:diagon of f squared}. Let $N:=(\max\{n^{-\beta/2},n^{-\delta_2},d^{-1/2}\})^{-1}$. Then for any vertices $v_i, v_j$ and $\varepsilon, \tau > 0$, there exists $C(\varepsilon, \tau)$ such that:

If $h^2$ is increasing:
\[
\mathbb{P}\left( \begin{array}{l}
(\mathbf{X}\mathbf{X}^\top)_{ii} \geq (\mathbf{X}\mathbf{X}^\top)_{jj} + N^{-1+\varepsilon},\\
\text{and } d_i/n \leq d_j/n + C N^{-(1-\varepsilon)/\beta}
\end{array} \right) \leq C(\varepsilon, \tau) n^{-\tau}.
\]

If $h^2$ is decreasing:
\[
\mathbb{P}\left( \begin{array}{l}
(\mathbf{X}\mathbf{X}^\top)_{ii} \leq (\mathbf{X}\mathbf{X}^\top)_{jj} + N^{-1+\varepsilon},\\
\text{and } d_i/n \geq d_j/n + C N^{-(1-\varepsilon)/\beta}
\end{array} \right) \leq C(\varepsilon, \tau) n^{-\tau}.
\]
\end{theorem}
\begin{proof}
Without loss of generality, assume that $h^2$ is increasing (the proof for the case where $h^2$ is decreasing is analogous).

Fix $\varepsilon, \tau > 0$. Consider the event
\[
    \Omega_{ij} := \left\{ (\mathbf{X}\mathbf{X}^\top)_{ii} > (\mathbf{X}\mathbf{X}^\top)_{jj} + N^{-1+\varepsilon} \right\}.
\]
By Theorem~\ref{thm:diag_concentration}, for all $l \in [n]$, w.h.p.,
\[
    (\mathbf{X}\mathbf{X}^\top)_{ll} = (h^2(\bbL))_{ll} + \ccalO(d^{-1/2}).
\]
Hence, by Lemma~\ref{lem:diagon of f squared}, on the event $\Omega_{ij}$ we have
\begin{align*}
    h^2(d_i/n)
    &\ge h^2(d_j/n) + N^{-1+\varepsilon} \\
        &\, + \ccalO\big(\max\{n^{-\beta/2}, n^{-\delta_2}, d^{-1/2}\}\big) \\
    &> h^2(d_j/n) + \tfrac{1}{2} N^{-1+\varepsilon}.
\end{align*}
Since $h^2$ is increasing and $\beta$-Hölder is continuous, it follows that
\[
    \frac{d_i}{n} \ge \frac{d_j}{n} + C N^{-(1-\varepsilon)/\beta}
\]
for some constant $C > 0$.  
This completes the proof.
\end{proof}

At a high-level, this result implies that if $h^2$ is increasing, then a larger diagonal entry of $\mathbf{X}\mathbf{X}^\top$ corresponds to a vertex with larger degree; and vice-versa if $h^2$ is decreasing.
From this we can immediately deduce a practical sampling rule:
if $h^2$ is \emph{monotone increasing}, we prioritize keeping nodes with \emph{largest scores}, where the scores are given by the diagonal entries of $\mathbf{X}\mathbf{X}^\top$; 
if $h^2$ is \emph{monotone decreasing}, we prioritize keeping nodes with \emph{lowest scores}, again defined by the diagonal of $\mathbf{X}\mathbf{X}^\top$. These rules guide the node selection direction used in feature-aware subsampling for trace preservation. 
The following corollary quantifies the benefit of this sampling rule compared to uniform node sampling.

\begin{corollary}[Trace preservation under monotone filters]
\label{cor:trace_preservation}
Under the assumptions of Theorem~\ref{thm:large diagonal is large degree}, let $\widetilde{\bbL}$ be the Laplacian of the subgraph sampled by $\mathrm{diag}(\mathbf{X}\mathbf{X}^\top)$, and $\cL_{\mathrm{u}}$ that from uniform sampling with the same ratio.

If $h^2$ is monotone increasing,
\[
\mathbb{P}\big(\mathrm{tr}(\widetilde{\bbL}) \ge \mathrm{tr}(\cL_{\mathrm{u}})\big) \ge 1 - C(\varepsilon,\tau)n^{-\tau},
\]
when keeping nodes with the largest diagonals of $\mathbf{X}\mathbf{X}^\top$.

If $h^2$ is monotone decreasing,
\[
\mathbb{P}\big(\mathrm{tr}(\widetilde{\bbL}) \ge \mathrm{tr}(\cL_{\mathrm{u}})\big) \ge 1 - C(\varepsilon,\tau)n^{-\tau},
\]
when keeping nodes with the smallest diagonals.
\end{corollary}

Thus, under monotone filters, feature-driven subsampling achieves stronger Laplacian trace preservation than uniform sampling at the same sampling ratio, w.h.p. In the next section, we empirically assess the practical impact of this result.

\begin{remark}[Non-monotone filters]
While the results in Theorem~\ref{thm:large diagonal is large degree} and Corollary~\ref{cor:trace_preservation} rely on the monotonicity of $h^2$, the proposed sampling framework remains robust in more general settings. If $h^2$ is non-monotone but exhibits \emph{local monotonicity} over the support of the graph's empirical spectral distribution, the ordering of $(\mathbf{X}\mathbf{X}^\top)_{ii}$ will still correlate with structural importance within that spectral band. In cases where $h^2$ is highly oscillatory—such as in certain band-pass filters—the diagonal entries of $\mathbf{X}\mathbf{X}^\top$ instead reflect the node's \emph{spectral centrality} relative to the filter's passband. In such scenarios, while the degree-ordering might be lost, the feature-driven sampling still targets nodes that maximize the preserved signal energy -- the ultimate goal of trace-preserving subsampling.
\end{remark}

\section{Experiments Under Coarse Alignment}

We evaluate the effectiveness of our sampling algorithm on the Planetoid citation networks Cora, CiteSeer, and PubMed, whose feature heterophilies are $0.0415$, $0.0194$, and $0.0007$, respectively~\cite{Yangetal2016}. Our objective is to compare the test accuracy achieved by GNNs trained on subgraphs sampled using our heuristic against those obtained via uniform random sampling, when both are evaluated on the full graph.

\subsection*{Trace Preservation}
To assess the structural quality of the sampled subgraphs, we examine the normalized Laplacian trace $\tilde{\mathrm{tr}}(\bbL) = \mathrm{tr}(\bbL) / n$, which serves as a surrogate for the rank of the Laplacian and thus relates directly to GNN expressivity. In accordance with the theoretical lower bound established in Proposition~\ref{prop:bound_trL}, Figure~\ref{fig:trace} shows that, for Cora and PubMed, our heuristic consistently produces subgraphs with larger average normalized trace than those produced by random sampling.

\subsection*{GNN Training and Transferability}

We investigate the transferability of GNNs trained on these sampled subgraphs. 
For each sampling budget (shown along the $x$-axis of Figure~\ref{fig:GNN Acc}), we train a GCN~\cite{kipf17-classifgcnn} on the sampled subgraph and evaluate its test accuracy on the full graph. All models use a two-layer GCN with hidden dimension $128$ and ReLU activations, trained for $300$ epochs. We apply the Adam optimizer~\cite{kingma17-adam} with learning rate $0.001$ and weight decay $0.0001$. All graphs are symmetrized prior to sampling to ensure undirected structure.

Figure~\ref{fig:GNN Acc} aggregates results from 50 independent runs. Error bars denote standard error, while red dots indicate the mean accuracy achieved using our sampling method. The largest performance improvements again occur on PubMed, consistent with its high feature homophily and the theoretical predictions derived in our analysis. Across all three datasets, our heuristic generally outperforms random node sampling, especially in the low- and high-budget regimes where retaining structural information is crucial.

\section{Experiments Under Fine Alignment}
\label{sec:experiments}

We empirically validate the fine alignment theory of Section~\ref{sec:fine_alignment} by studying both synthetic data generated from 
\[
    \mathbf{X} = h(\bbL) \bbX_0,
\]
and real-world graph datasets.  
Our goals are: (i) to verify the trace-preservation predictions of Corollary~\ref{cor:trace_preservation}, 
and (ii) to assess the implications for downstream GNN transferability.

\subsection{Synthetic Experiments on Stationary Graph Signals}
\label{subsec:synthetic-fine}
\begin{figure}[t]
    \centering
    \includegraphics[width=1\linewidth]{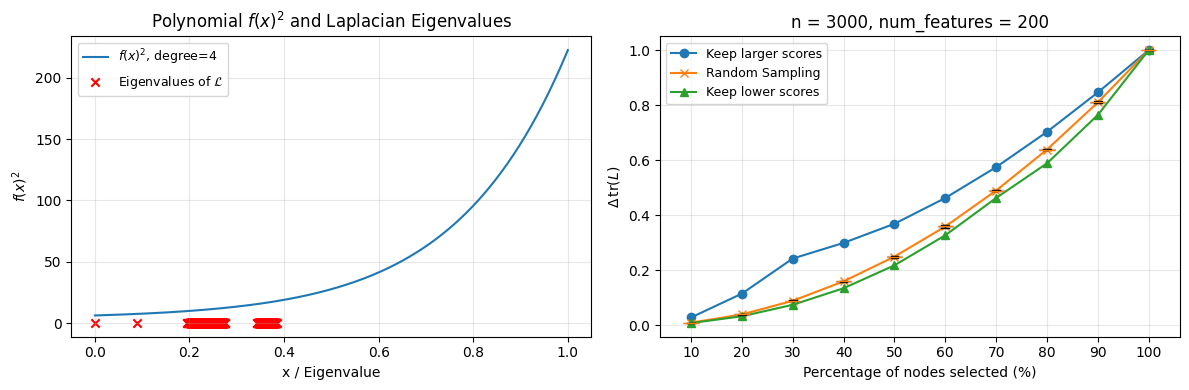}
    \includegraphics[width=1\linewidth]{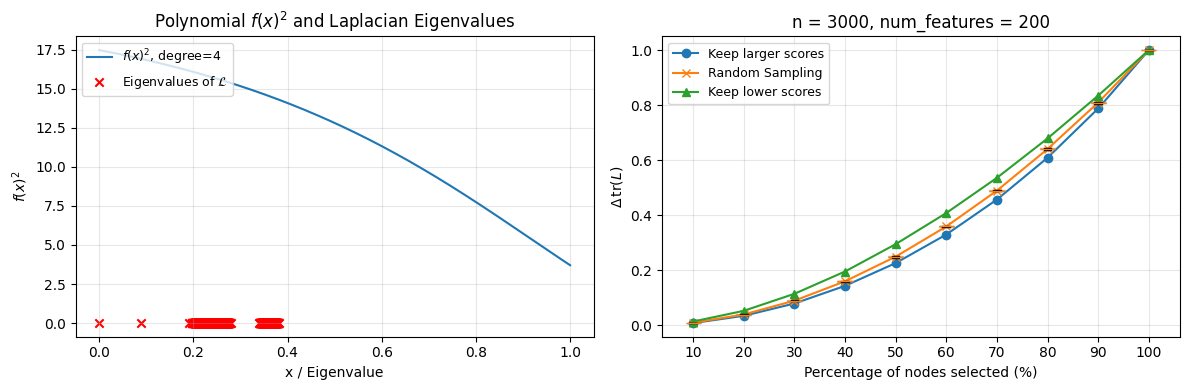}
    \includegraphics[width=1\linewidth]{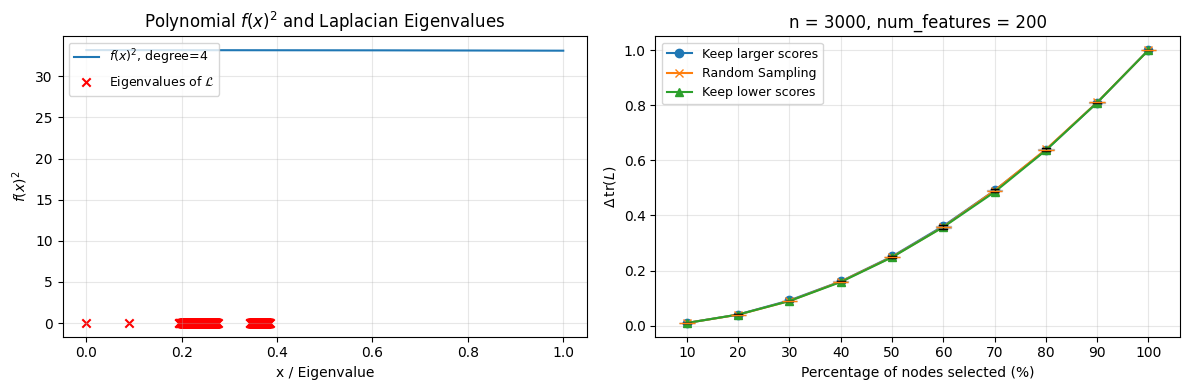}
    \caption{
        Synthetic stationary-graph experiments under monotone and flat $h^2$. The red dots represent the distribution of the eigenvalues of the rescaled Laplacian $\tbL$.
    }
    \label{fig:synthetic-alignment}
\end{figure}

\begin{figure*}[t!]
    \centering
    \includegraphics[width=0.32\linewidth]{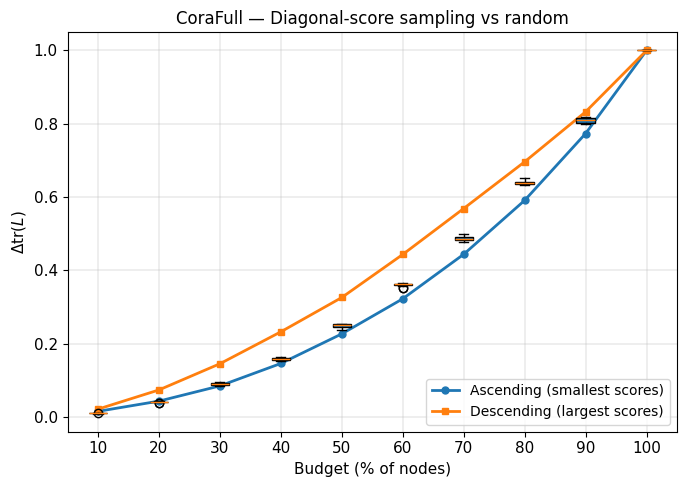}
    \includegraphics[width=0.32\linewidth]{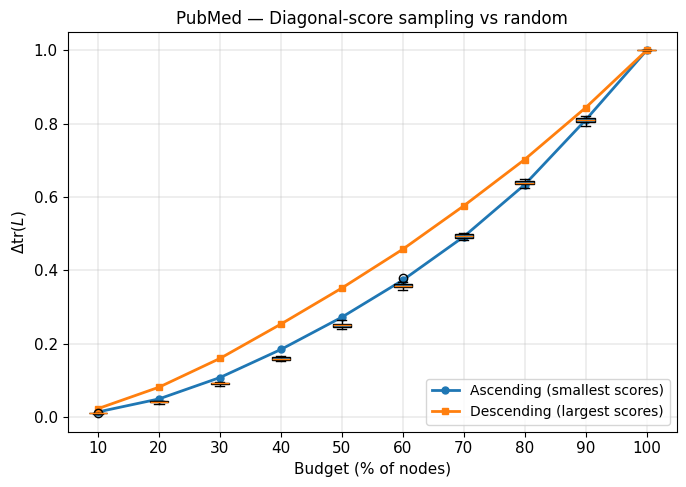}
    \includegraphics[width=0.32\linewidth]{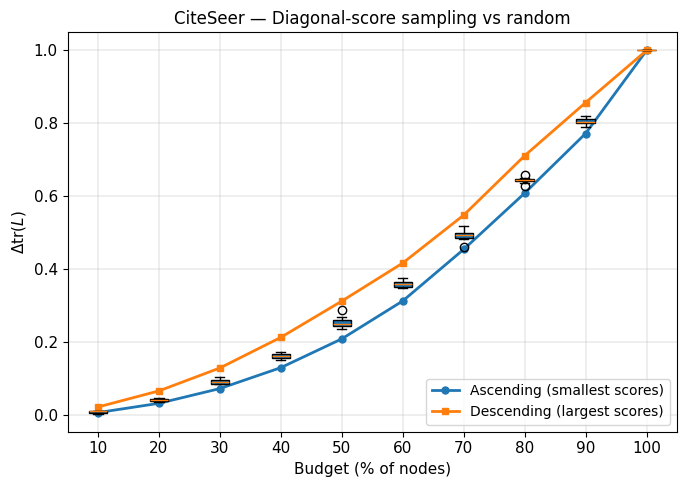}
    \includegraphics[width=0.32\linewidth]{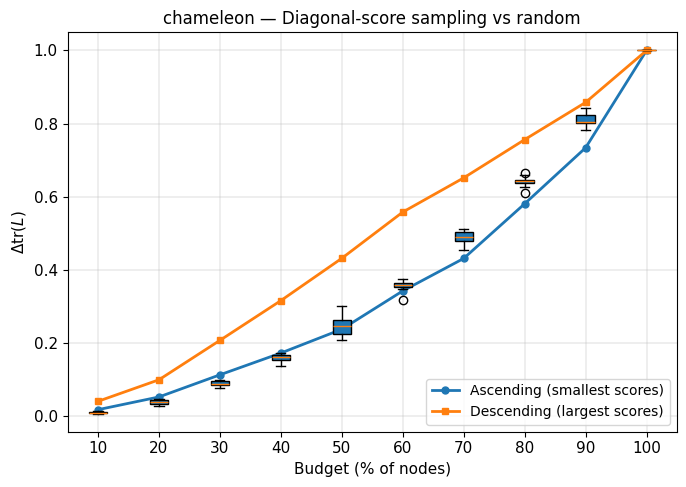}
    \includegraphics[width=0.32\linewidth]{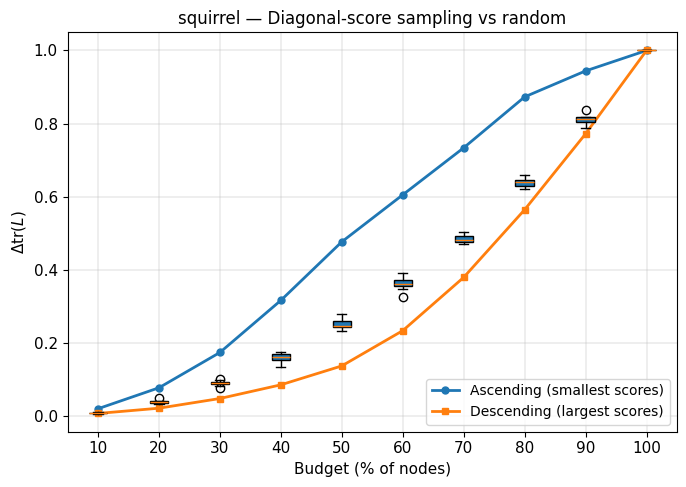}
    \includegraphics[width=0.32\linewidth]{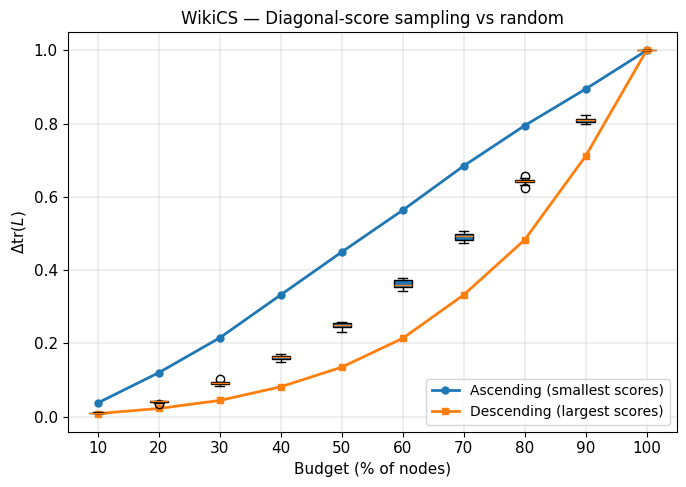}   
 
    \caption{
        Summary of trace-preservation performance $\Delta\mathrm{tr}(\tbL)$ on real-world datasets.
        For each dataset, one of the two rules (keep-largest or keep-smallest)
        uniformly dominates across all sampling ratios $\gamma$.
    }
    \label{fig:synth-monotone}
\end{figure*}

We first examine the stationary model in controlled settings.  
Graphs are generated from a stochastic block \emph{graphon} model, instantiated as finite random graphs with rescaled Laplacian $\tbL$.  
Given $\tbL$, we draw i.i.d.\ Gaussian features $X_0$ and construct $X = h(\tbL) X_0$ using random polynomials $h$ of degree at most five.  
For each node $i$, we compute the diagonal alignment score
\[
    s_i = (\mathbf{X}\mathbf{X}^\top)_{ii},
\]
and compare two node-retention rules: keep-largest and keep-smallest $s_i$, under a fixed sampling budget. 

Performance is measured by the trace preservation ratio
\[
    \Delta \mathrm{tr}(\tbL)
    = \frac{\mathrm{tr}(\tbL_{\mathrm{sub}})}{\mathrm{tr}(\tbL)},
\]
where $\tbL_{\mathrm{sub}}$ is the Laplacian of the sampled subgraph. As shown in Figure~\ref{fig:synthetic-alignment}, the trace preservation behavior depends on the monotonicity of the filter response $h^2$.

The empirical results perfectly match theory:

\begin{itemize}
    \item If $h^2$ is \emph{monotone increasing} on the spectral support of $\tbL$, keeping \emph{larger} $s_i$ yields the largest $\Delta\mathrm{tr}(\tbL)$.
    \item If $h^2$ is \emph{monotone decreasing}, keeping \emph{smaller} $s_i$ maximizes trace preservation compared to keeping largest.
\end{itemize}

This monotonicity rule holds uniformly across graphon realizations, sampling budgets, and random polynomials whose squared response is monotone on the spectrum.
When $h^2$ is non-monotone, no single rule is uniformly optimal, consistent with the absence of a universal monotone ordering of spectral energies.

\subsection{Experiments on Real Graph Data}
\label{subsec:real-fine}

We next evaluate the sampling rules on real-world graphs using the stationary alignment score $s_i = (\mathbf{X}\mathbf{X}^\top)_{ii}$ computed from \emph{centered} node features (as the stationary model implies zero-mean signals).  

\begin{figure*}[t]
    \centering
    \includegraphics[width=0.32\linewidth]{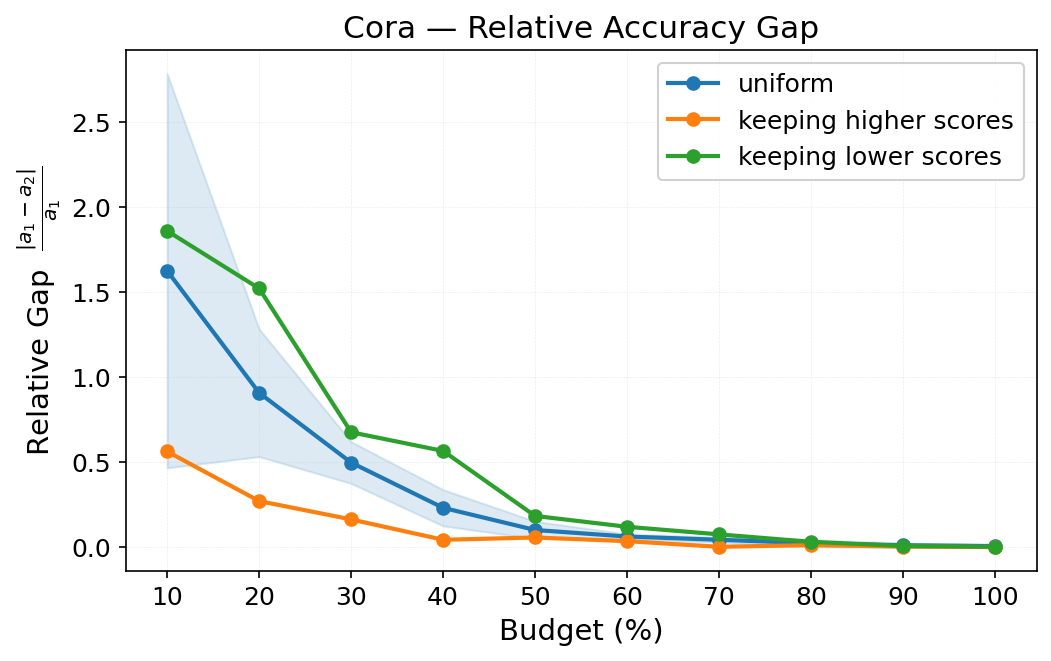}
    \includegraphics[width=0.32\linewidth]{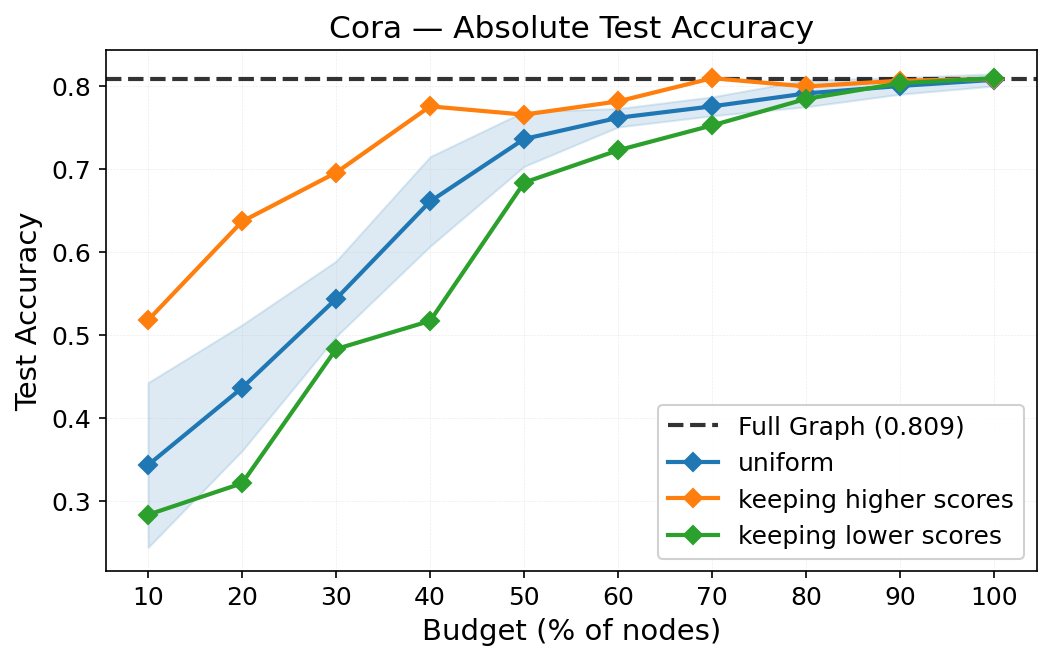}
    \includegraphics[width=0.32\linewidth]{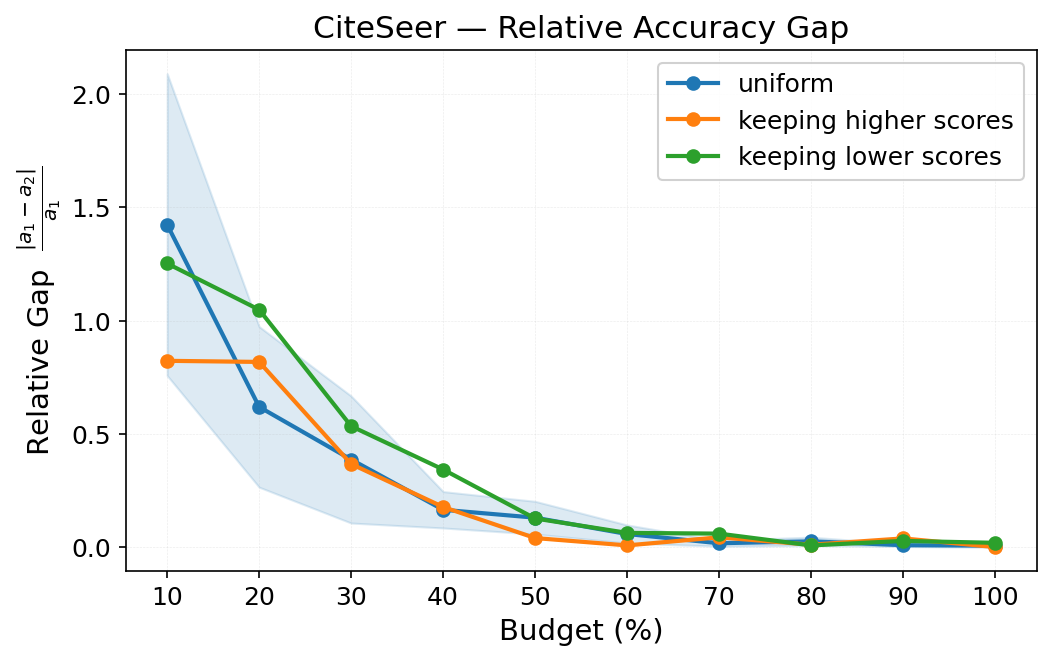}
    \includegraphics[width=0.32\linewidth]{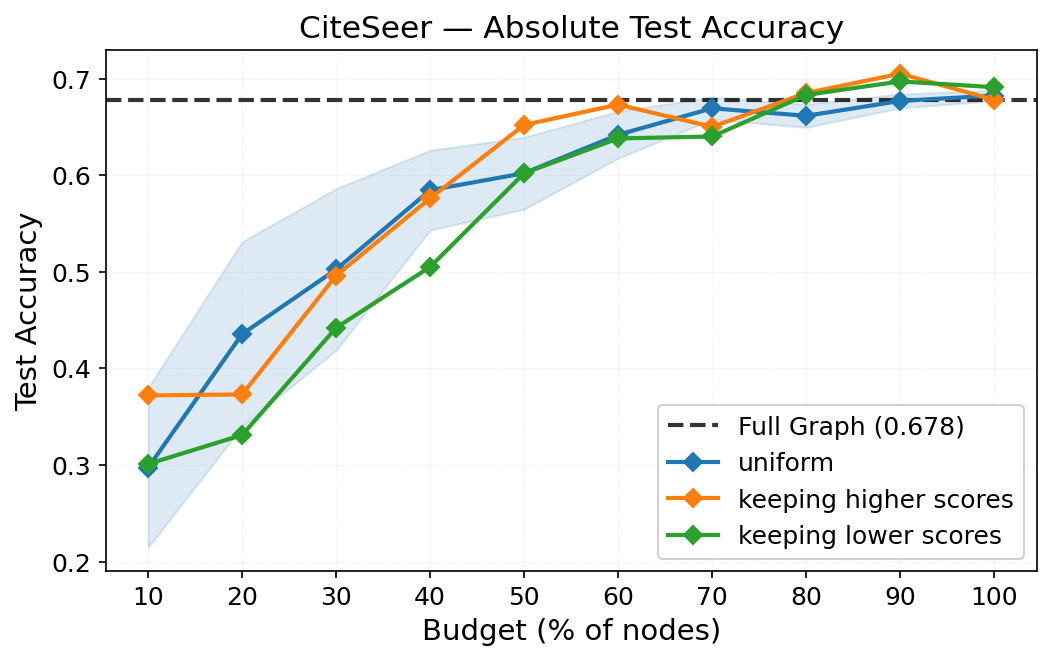}
    \includegraphics[width=0.32\linewidth]{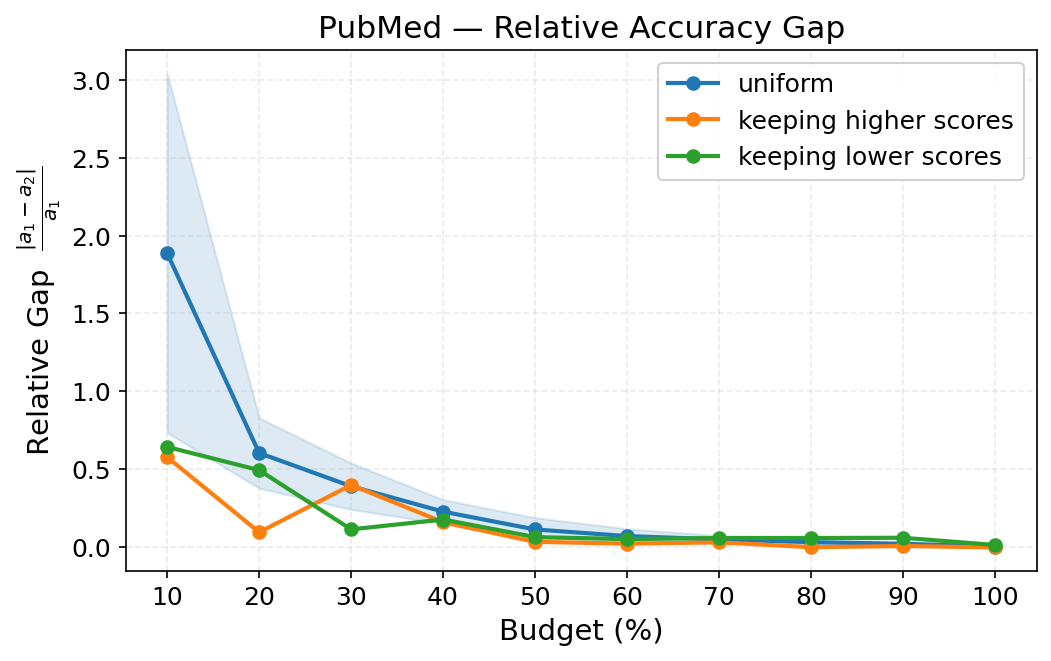}
    \includegraphics[width=0.32\linewidth]{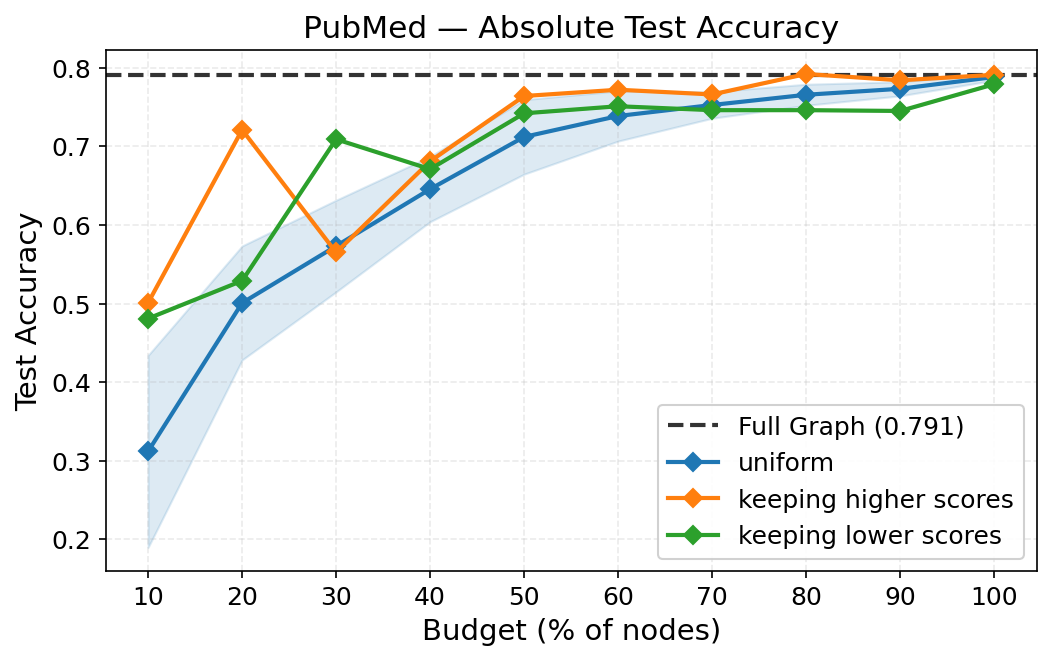}
    
    \caption{
        GNN transferability on sampled subgraphs.
        The retention rule yielding higher $\Delta\mathrm{tr}(\tbL)$ consistently exhibits
        smaller transfer error. Even though the superior performance is not uniformly held, we observe that the algorithm that behaves better is the algorithm that corresponds to the one that is more preserving of the trace of the Laplacian.
    }
    \label{fig:gnn-transfer}
\end{figure*}

\subsection*{Trace Preservation}

For each dataset and sampling ratio~$\gamma$, we compute $\Delta\mathrm{tr}(\tbL)$ under both keep-largest and keep-smallest scores. Across citation networks, social networks, and other benchmarks with varying label and feature homophily, from Figure~\ref{fig:synth-monotone} we observe that on sufficiently large graphs, {one} of the two rules consistently dominates the other, and random subsampling, over {all} budgets.


\subsection*{GNN Training and Transferability}

To study downstream performance, we train a two-layer GCN on sampled subgraphs and evaluate the trained model {directly on the full graph} without fine-tuning, the standard transferability setting.  
Let $a_2$ denote accuracy when training on the full graph, and let $a_1(\gamma)$ denote the test accuracy for the model trained only on the sampled subgraph but tested on the test nodes of the whole graph. 
We measure the relative transfer gaps
\[
    \frac{|a_1(\gamma)-a_2|}{a_1(\gamma)} \text{.}
\]

Across all datasets, the node-retention rule that yields {higher} $\Delta\mathrm{tr}(\tbL)$ also yields {smaller} transfer gaps, with the largest gains observed at small budgets where uniform sampling severely degrades connectivity and trace.  
These findings confirm that the fine alignment model correctly identifies nodes that preserve the effective spectral rank, and that this preservation directly improves GNN scalability and transfer performance.

\section{Discussion}
\label{sec:discussion}

From a signal processing viewpoint, the behavior of the spectral filter \( h \) on the graph Laplacian eigenvalues is central in determining which frequency components are retained or suppressed during learning. In our analysis, we are particularly interested in the implications of the monotonicity of \( h \), on how this relates to the interaction of the feature correlation matrix with the graph, and on how to determine the appropriate node sampling strategy.

\subsection{Monotonicity and Filter Energy}

Let us consider a spectral filter \( h(\lambda) \) applied to Laplacian eigenvalues \( \lambda \). If \( h \) is monotone increasing, then its square \( h^2(\lambda) \), which corresponds to the power spectral response, is also monotone increasing. Similarly, if \( h \) is monotone decreasing, then \( h^2 \) preserves the decreasing trend. This preservation of monotonicity is essential as it determines which spectral components of $\bbX$ in \eqref{def:stationary_model} have their energy preserved.
A monotone increasing \( h^2 \) will favor the amplification of high-frequency components in the graph spectrum, while attenuating the low-frequency ones. Conversely, a decreasing \( h^2 \) suppresses high frequencies and emphasizes low frequencies.

\subsection{Connections to Graph Homophily and Heterophily}

This frequency selection has important implications for the underlying structure of the graph signal. In graphs exhibiting \emph{feature homophily}, neighboring nodes tend to share similar features, resulting in smooth graph signals. These signals have dominant energy in the low-frequency spectrum, i.e., the components associated with small eigenvalues of the Laplacian. Therefore, effective learning in homophilic settings typically relies on \textbf{low-pass filters}, where \( h \) (and hence \( h^2 \)) is monotone decreasing.

In contrast, \emph{heterophilic} graphs contain feature-discontinuous signals where neighboring nodes often have dissimilar features. Such signals are non-smooth and contain significant high-frequency content. As a result, learning tasks on heterophilic graphs benefit from \textbf{high-pass filters}, corresponding to monotone increasing \( h \) and \( h^2 \), which accentuate high-frequency components.

\subsection{Implications for Node Sampling}

The alignment between the monotonicity of \( h^2 \) and the graph's feature correlation provides a principled guide for node selection. When \( h^2 \) is increasing---corresponding to a heterophilic regime---it is beneficial to retain nodes with higher energy scores (e.g., larger row norms of \( \mathbf{X} \)). Conversely, when \( h^2 \) is decreasing---aligned with homophilic graphs---it is optimal to retain nodes with lower energy scores to preserve the dominant low-frequency content.

This relationship underlies the sampling strategies introduced in Section \ref{sec:fine_alignment}, which adapts the node selection policy based on the monotonicity of the target filter response. Such sampling not only aligns with the energy distribution of the signal but also enhances generalization by focusing on structurally informative nodes.

\section{Conclusions}
\label{sec:conclusions}

In this work, we proposed a unified framework for improving the transferability of GNNs through principled subgraph sampling, grounded in both theoretical foundations and empirical validation. Our approach is centered around two alignment models:
coarse alignment, which introduces a novel notion of feature homophily and proposes a trace-preserving sampling algorithm based on the graph Laplacian and the normalized feature correlation matrix; and fine alignment, which refines the notion of alignment through a stationary signal model demonstrating that node-wise feature energy (i.e., the diagonal of the covariance matrix) concentrates around spectral responses of the graph filter. This connection justifies the use of spectral energy-based node scoring and guides the optimal sampling direction based on the monotonicity of the filter response.
We derived rigorous theoretical guarantees under the stationary model and the SBM, showing that our sampling strategy preserves the trace of the Laplacian and enhances the expressivity of GNNs. Extensive experiments on both synthetic and real-world graphs validate our theoretical insights, demonstrating superior performance of the proposed method over uniform sampling in terms of both trace preservation and GNN transferability.
These findings offer a deeper understanding of how graph structure and feature alignment jointly impact the generalization ability of GNNs across scales. They also provide practical tools for scalable GNN training via sampling, paving the way for broader applications in large-scale graph learning.

\bibliographystyle{IEEEtran}
\bibliography{myIEEEabrv,strings,refs,bib_cumulative}
\appendix
\section{Proof of Lemma~\ref{lem:spectral_op_concentration}}
\label{app:proof_spectral_concentration}
To prove Lemma \ref{lem:diagon of f squared}, we first invoke the following lemma. It allows to transfer the analysis of the random graph Laplacian $\bbL$ to that of the expected rescaled Laplacian $\overline{\bbL}$, which makes the analysis easier since we know the structure of $\overline{\bbL}$ well.
\begin{lemma}[Concentration of the Laplacian operator]
\label{lem:op_concentration}
Let $G$ be a random graph on $n$ vertices, with independent edge probabilities $\{\bbP_{ij}\}_{i > j}$. Then, for all $\tau > 0$, there exists a constant $C > 0$ such that
\[
 \| \bbL - \overline{\bbL} \|_{\mathrm{op}} \leq C n^{-1/2}
\]
with probability at least $1 - Cn^{-\tau}$. Moreover, if there exist constants $\delta_1, \delta_2 > 0$ such that
\[
 \sum_{i > j} \bbP_{ij}(1 - \bbP_{ij}) \geq \delta_1 n^{1 + \delta_2},
\]
then we also have
\[
 \| \bbL - \overline{\bbL} \|_{\mathrm{op}} \geq C^{-1} n^{-1 + \delta_2 / 2}
\]
with probability at least $1-Cn^{-\tau}$.
\end{lemma}
\begin{proof}
We write the normalized Laplacian as $\mathbf{L} = n^{-1} (\mathbf{D} - \mathbf{A})$, with $\mathbf{A}_{ij} \sim \mathrm{Bernoulli}(\bbP_{ij})$ for $i > j$ and $\mathbf{A}_{ii} = 0$. The upper bound follows by applying matrix Bernstein to $\mathbf{A}$ and Chernoff bounds to $\mathbf{D}$.

For the lower bound, we first observe:
\[
\| \bbL - \overline{\bbL} \|_{\mathrm{op}}
\geq n^{-1/2} \| \bbL - \overline{\bbL} \|_F
\geq n^{-3/2} \sqrt{ 2 \sum_{i>j} (\bbA_{ij} - \overline{\bbA}_{ij})^2 }.
\]
Define $\bbX_{ij} := (\bbA_{ij} - \overline{\bbA}_{ij})^2$, for which $|\bbX_{ij}| \leq 4$, $\mathbb{E}[\bbX_{ij}] = \bbP_{ij}(1 - \bbP_{ij})$, and $\operatorname{Var}(\bbX_{ij}) \leq 1$. Bernstein's inequality yields
\[
\sum_{i>j} \bbX_{ij}
\geq
\sum_{i>j} \bbP_{ij}(1 - \bbP_{ij}) - \ccalO(n),
\]
which exceeds $\delta_1 n^{1 + \delta_2}/2$ w.h.p. This completes the proof.
\end{proof}

Lemma \ref{lem:op_concentration} then allows proving the following spectral concentration result for $h^2(\bbL)$ under H\"older continuity.

\begin{lemma}[Spectral concentration under Hölder smoothness]
\label{lem:spectral_op_concentration}

Let $G$ be a random graph on $n$ nodes as above. For any $\tau > 0$, there exists a constant $C > 0$ such that the following holds. Let $g: [0,2] \to \mathbb{R}$ be $\beta$-Hölder continuous with $0 < \beta < 1$. Then
\[
 \| g(\mathbf{L}) - g(\overline{\mathbf{L}}) \|_{\mathrm{op}} \leq C_g n^{-\beta/2}
\]
with probability at least $1 - C n^{-\tau}$. If additionally,
\[
 \sum_{i > j} \bbP_{ij}(1 - \bbP_{ij}) \geq \delta_1 n^{1 + \delta_2},
\]
for some $\delta_1, \delta_2 > 0$, and $g$ is Lipschitz, then
\[
 \| g(\mathbf{L}) - g(\overline{\mathbf{L}}) \|_{\mathrm{op}}
 \leq C_g \log(n)\, n^{-1/2}
\]
with the same probability.
\end{lemma}

\begin{proof}
The proof follows directly from Lemma \ref{lem:op_concentration} and standard operator Hölder–Zygmund results, specifically (1.1), (1.2), and Corollary 7.5 of \cite{aleksandrov2010operator_holder_zygmund}
\end{proof}

Finally, we also need the following lemma, which describes the spectral structure of the expected Laplacian $\Bar{\bbL}$.

\begin{lemma}[Spectral Structure of Expected Laplacian]
\label{lem:eigenvalues/vectos of expected lapplacian}

Suppose $\bbL$ is sampled from an SBM $\mathcal{G}(\boldsymbol{n}, \bbP)$. Then the expected Laplacian is the following $n$ by $n$ matrix:
\[
n \cdot \overline{\bbL} =
\begin{bmatrix}
    -\bbP_{11} \mathbf{1}_{n_1} \mathbf{1}_{n_1}^\top + w_1 \bbI_{n_1} & \cdots & -\bbP_{1r} \mathbf{1}_{n_1} \mathbf{1}_{n_r}^\top \\
    \vdots & \ddots & \vdots \\
    -\bbP_{r1} \mathbf{1}_{n_r} \mathbf{1}_{n_1}^\top & \cdots & -\bbP_{rr} \mathbf{1}_{n_r} \mathbf{1}_{n_r}^\top + w_r \bbI_{n_r}
\end{bmatrix},
\]
where $w_s = \sum_{t=1}^r n_t \bbP_{st}$ and $\mathbf{1}_{n_s}$ is the all-ones vector in $\mathbb{R}^{n_s}$. Moreover, define the matrix
\[
    \bbM = \mathrm{diag}(w_1, \ldots, w_r) - \bbP \cdot \mathrm{diag}(\boldsymbol{n}),
\]
which is diagonalizable since $[\mathrm{diag}(\boldsymbol{n})]^{1/2}\bbM[\mathrm{diag}(\boldsymbol{n})]^{-1/2}$ is real and symmetric. Let $\mu_1 = 0, \mu_2, \ldots, \mu_r$ be the eigenvalues of $\bbM$ with eigenvectors $\mathbf{u}_1 = \mathbf{1}_r, \ldots, \mathbf{u}_r \in \mathbb{R}^r$. Let $\{ \mathbf{v}_t^{(s)} \}_{t=2}^{n_s}$ be orthonormal vectors completing $\mathbf{v}_1^{(s)} = n_s^{-1/2} \mathbf{1}_{n_s}$ to a basis of $\mathbb{R}^{n_s}$.

Then, the eigenvalues and eigenvectors of $n \cdot \overline{\bbL}$ are:
\begin{itemize}
    \item $\mu_1, \ldots, \mu_r$ corresponding to eigenvectors
    \[
    \left\{
    \frac{1}{\sqrt{ \sum_{j=1}^r u_s(j)^2 n_j }}
    \begin{bmatrix}
        u_s(1) \mathbf{1}_{n_1} \\
        \vdots \\
        u_s(r) \mathbf{1}_{n_r}
    \end{bmatrix}
    \right\}_{s=1}^r.
    \]
    \item $w_s$ with multiplicity $n_s - 1$ for $s = 1, \ldots, r$, corresponding to block-orthogonal vectors of the form:
    \[
        \left\{
        \begin{bmatrix}
            \v^{(1)}_t\\
            0\\
            \vdots\\
            0
        \end{bmatrix}
        \right\}_{t=2}^{n_1},
        \ldots,
        \left\{
        \begin{bmatrix}
            0\\
            \vdots\\
            0\\
            \v_t^{(r)}
        \end{bmatrix}
        \right\}_{t=2}^{n_r}.
    \]
\end{itemize}
\end{lemma}

\begin{proof}
The proof is simple and we omit it for brevity.
\end{proof}

We are now ready to prove Theorem Lemma ~\ref{lem:diagon of f squared}.

\begin{proof}[Proof of Lemma \ref{lem:diagon of f squared}]
By Lemma \ref{lem:spectral_op_concentration}, we have w.h.p.:
 \[
|[h^2(\bbL)]_{ii}-[h^2[\Bar{\bbL}]]_{ii}|\leq\|h^2(\bbL)-h^2(\Bar{\bbL})\|_\op\leq C\log(n)\times n^{-\beta/2}
 \]
Let $g:[n]\rightarrow[r]$ be the membership assignment function. By Lemma \ref{lem:eigenvalues/vectos of expected lapplacian}, we can write, for any $i\in[n]$,
\begin{align*}
(h^2(\Bar{\bbL}))_{ii}&=\sum_{s=1}^r\frac{h^2(\mu_s/n)\mathbf{u}_s(g(i))^2}{\mathbf{u}_s(1)^2n_1+\cdots+\mathbf{u}_s(r)^2n_r}\\
&+h^2(w_{g(i)}/n)\sum_{t=2}^{n_{g(i)}}\mathbf{v}_t^{(g(i))}(i)^2\\
 &=\sum_{s=1}^r\frac{h^2(\mu_s/n)\mathbf{u}_s(g(i))^2}{\mathbf{u}_s(1)^2n_1+\cdots+\mathbf{u}_s(r)^2n_r}\\
 &+h^2(w_{g(i)}/n)\left(1-\frac{1}{n_{g(i)}}\right)
\end{align*}
 where in the last line we used that $(n_s)^{-1/2}\mathbf{1}_{n_s},\mathbf{v}_2,...,\mathbf{v}_{n_s}$ forms an orthonormal basis of $\R^{n_s}$ for any $1\leq s\leq r$. From the closed form expression of $\bbM$ and a Frobenius norm bound one can easily deduce $n^{-1}\|\bbM\|_\op\leq Cr$ for some constant $C>0$, so we have $|h^2(\mu_s/n)|\leq C'r^\beta$. Applying the lower bound $\min\{n_s\}\geq n^{\delta_2}$ gives $\mathbf{u}_s(1)^2n_1+\cdots+\mathbf{u}_s(r)^2n_r\geq n^{\delta_2}$. Hence,
 \begin{align*}
&\left|\sum_{s=1}^r\frac{h^2(\mu_s/n)\mathbf{u}_s(g(i))^2}{\mathbf{u}_s(1)^2n_1+\cdots+\mathbf{u}_s(r)^2n_r}\right|\\
 &\leq \frac{r^{1+\beta}}{n^{\delta_2}}\leq\frac{(\log(n))^{\delta_1(1+\beta)}}{n^{\delta_2}}=\ccalO(n^{-\delta_2}).
\end{align*}
 Similarly, $w_s/n\leq r$ gives $h^2(w_s/n)\leq Cr^\beta$ so we have
 \[
 h^2(w_{g(i)}/n)\left(1-\frac{1}{n_{g(i)}}\right)=h^2(w_{g(i)}/n)+\ccalO(n^{-\delta_2}).
 \]
 On the other hand for any $i\in[n]$, we have $w_{g(i)}/n-\bbP_{g(i)g(i)}/n=(\Bar{\bbL})_{ii}$. And by definition $d_i/n=(\bbL)_{ii}$, so applying Lemma \ref{lem:op_concentration} again gives
 \begin{align*}
 (h^2(\bbL))_{ii}&=h^2(w_{g(i)}/n)+\ccalO(n^{-\beta/2})+\ccalO(n^{-\delta_2})\\
 &=h^2(d_i/n)+\ccalO(\max\{n^{-\beta/2},n^{-\delta_2}\})
  \end{align*}
 This completes the proof.
\end{proof}

\end{document}